\documentclass[11pt]{article}

\usepackage{booktabs} 

\usepackage{bbm}
\usepackage{enumitem}
\usepackage{times}
\usepackage{fullpage}

\usepackage{amsmath}
\usepackage{amsfonts}
\usepackage{amsthm}
\usepackage[utf8]{inputenc}
\usepackage[english]{babel}

\newtheorem{theorem}{Theorem}
\newtheorem{corollary}{Corollary}
\newtheorem{lemma}{Lemma}

\usepackage{authblk}
\allowdisplaybreaks

\usepackage{amsmath,amsthm,amssymb,dsfont}
\allowdisplaybreaks
\usepackage{tikz}
\usepackage[noend]{algorithmic}
\usepackage{algorithm,float}
\usepackage{enumitem}
\usepackage[numbers]{natbib}

\newlist{legal}{enumerate}{10}
\setlist[legal]{label*=\arabic*.}

\newtheorem{claim}{Claim}

\newcommand{\one}{\ensuremath{\mathds{1}}}
\newcommand{\ct}{\ensuremath{\mathsf{count}}}
\newcommand{\rej}{\ensuremath{\mathsf{rej}}}

\begin{document}
\title{Online Non-Preemptive Scheduling to Minimize Weighted Flow-time on Unrelated Machines}

\author[1]{Giorgio Lucarelli \thanks{giorgio.lucarelli@imag.fr}}
\author[2]{Benjamin Moseley \thanks{moseleyb@andrew.cmu.edu}}
\author[3]{Nguyen Kim Thang \thanks{thang@ibisc.fr}}
\author[4,5]{Abhinav Srivastav \thanks{abhinav.srivastav@ens.fr}}
\author[1]{Denis Trystram \thanks{trystram@imag.fr\\ \\}}

\affil[1]{Universit\'e GrenobleAlpes, CNRS, INRIA, Grenoble INP, LIG\\}
\affil[2]{Carnegie Mellon University\\}
\affil[3]{Universit\'e  d'\'Evry, Universit\'e Paris-Saclay\\}
\affil[4]{D\'epartment d'Informatique, ENS-Paris\\}
\affil[5]{LAMSADE, Universit\'e Paris Dauphine}

\maketitle

\begin{abstract}

In this paper, we consider the online problem of scheduling independent jobs \emph{non-preemptively} so as to minimize the weighted flow-time on a set of unrelated machines. There has been a considerable amount of work on this problem in the preemptive setting where several competitive algorithms are known in the classical competitive model. 
However, the problem in the non-preemptive setting admits a strong lower bound.
Recently, Lucarelli et al. presented an algorithm that achieves a $O\left(\frac{1}{\epsilon^2}\right)$-competitive ratio when the algorithm is allowed to reject $\epsilon$-fraction of total weight of jobs and $\epsilon$-speed augmentation.
They further showed that speed augmentation alone is insufficient to derive any competitive algorithm. 
An intriguing open question is whether there exists a scalable competitive algorithm that rejects a small fraction of total weights.

In this paper, we affirmatively answer this question. Specifically, we show that there exists a $O\left(\frac{1}{\epsilon^3}\right)$-competitive algorithm for minimizing weighted flow-time on a set of unrelated machine that rejects at most $O(\epsilon)$-fraction of total weight of jobs. The design and analysis of the algorithm is based on the primal-dual technique. Our result asserts that alternative models beyond speed augmentation should be explored when designing online schedulers in the non-preemptive setting in an effort to find provably good algorithms.  

 \end{abstract}

\section{Introduction}

In this work, we study the fundamental problem of online scheduling of independent jobs on unrelated machines. Jobs arrive over time and the online algorithm has to make the decision which job to process \emph{non-preemptively} at any time on each machine. A job $j$ is released at time $r_{j}$ and takes $p_{ij}$ amount of processing time on a machine $i$. Further, each job has a weight $w_j$ that denote its (relative) priority. 
Our aim is to design a non-preemptive schedule that minimizes the total weighted flow-time (or response time) quantity, i.e., $ \sum_{j} w_j (C_j - r_j)$ 
where $C_j$ denote the completion time of job $j$. 

We are interested in designing online non-preemptive scheduling problem in the worst-case model. Several strong lower bounds are known for simple instances \cite{Bansal:2009,ChekuriKhanna01:Algorithms-for-minimizing}. The main hurdle arises from two facts: the algorithm must be online and robust to all problem instances and the algorithmic decisions made should be of irrevocable nature. In order to overcome the strong theoretical lower bound, Kalyanasundaram and Pruhs~\cite{KalyanasundaramPruhs00:Speed-is-as-powerful} and Phillips et al.~\cite{PhillipsStein02:Optimal-time-critical} proposed the analysis of scheduling algorithms in terms of the speed augmentation and machine augmentation, respectively. Together these augmentation are commonly referred to as resource augmentation. Here, the idea is to either give the scheduling algorithm faster processors or extra machines in comparison to the adversary. For preemptive problems, these models provide a tool to establish a theoretical explanation for the good performance of algorithms in practice. In fact, many practical heuristics have been shown to be competitive where the algorithm is given resource augmentation. In contrast, problems in the non-preemptive setting have resisted against provably good algorithms even with such additional resources \cite{LucarelliThang16:Online-Non-preemptive}.

Choudhury et al.~\cite{ChoudhuryDas15:Rejecting-jobs} proposed a new model of resource augmentation where the online algorithm is allowed to reject some of the arriving jobs, while the adversary must complete all jobs. Using a combination of speed augmentation and rejection, Lucarelli et al.~\cite{LucarelliThang16:Online-Non-preemptive} break this theoretical barrier and gave a scalable algorithm for non-preemptive weighted flow-time problems. However, it remains an intriguing question about the power of rejection model in comparison to the previous ones. 

Recently,  Lucarelli et al. \cite{LucarelliThang16:Unweighted-Non-preemptive} showed that a $O(1)$ competitive algorithm exists if all jobs have unit weight and one only rejects a constant fraction of the jobs. Their algorithm and analysis are closely tied to the unweighted case and there is no natural extension to the case where jobs have weights.  The  question looms, does there exist  a constant competitive algorithm for non-preemptive scheduling to minimize weighted flow-time using rejection? 

\subsection{Our Result and Approach}

This paper gives the first algorithm with non-trivial guarantees for minimizing weighted flow time using rejection and no other form of resource augmentation.  The main result of the paper is the following theorem. The theorem shows that constant competitiveness can be achieved by only rejecting a small faction of the total weight of the jobs.

\begin{theorem} \label{main-theorem}
For the non-preemptive problem of minimizing weighted flow-time on unrelated machines, 
there exists a $O(\frac{1}{\epsilon^3})$-competitive algorithm that rejects at most $O(\epsilon)$-fraction of total weight of the jobs for any $0 < \epsilon < 1$.
\end{theorem}


The algorithmic decisions are classified into three parts: dispatching, rejecting and scheduling policy. 
The scheduling follows HDF policy (Highest Density First) once jobs are assigned to the machines. At the arrival of a job, for each machine, the algorithm computes an approximate increase in the weighted flow-time and assigns the job to the machine with the least increase in the approximate weighted flow-time. To compute this quantity for a given machine, the algorithm considers the set of uncompleted jobs in the machine queue in the non-increasing order of densities 
and uses two different rejection policies. 

The first rejection policy, referred as the \emph{preempt} rule, rejects jobs that have already started processing 
if the total weight of newly arrived  ``high priority'' jobs (high density jobs) exceeds a given threshold. 
Specifically, when a job starts executing, we associate a counter that keep tracks of the total weight of newly arrived jobs. 
Once the value of this counter is at least $1/\epsilon$ times the weight of the current executing job, 
the algorithm preempts the current executing job and rejects it. 
The rejected job is pushed out of the system so as to be never executed again.


We emphasize here a critical issue due to job rejection which is of different nature to speed augmentation.
Observe that rejecting a job that has already started processing may cause a large decrease in the
weighted flow-time of the jobs in the machine queue. Due to job arrivals and job rejections, 
quantities associated to the machine queue (for example the remaining job weight, etc) vary arbitrarily without any nice properties like monotonicity.
That creates a significant challenge in the dual fitting analysis. 
To tackle this problem, we introduce the notion of \emph{definitive completion time} for each job. 
Once a job is rejected or completed before its definitive completion time, the algorithm removes the job from the queue of the machine. 
However, for the purpose of analysis, the rejected jobs are still considered in the definition of dual variables  until their definitive completion time. 
This ensures that for any fixed time, the weight of jobs not yet definitively completed 
increases with the arrival of new jobs (see Section~\ref{dual variables} for details). 

The second rejection policy, referred as the \emph{weight-gap} rule, rejects unprocessed ``low priority'' jobs
(small density jobs) from the machine's queue. This policy simulates the $\epsilon$-speed-augmentation. 
In the particular case where all jobs have the same weight, 
this rejection policy rejects a ``low priority'' job for every $1/\epsilon$ arrivals of new jobs. 
Due to the scheduling policy, if a ``low priority'' job is not rejected, then it will be completed last in the schedule (assuming no future job arrivals). 

In the algorithm's schedule, future arriving jobs do not delay the rejected low priority jobs, while the later ones need to be completed in the adversary's schedule.  This is where the algorithm benefits from the power of rejection. Specifically, the algorithm can use the difference between the rejection time and the definitive completion time of jobs to create a similar effect to speed augmentation. The key idea is to reject the low priority jobs so their total weight is comparable to jobs that arrive after them. 

The definitive completion times play crucial role so that the dual achieves a substantial value compared to the primal. By carefully choosing 
the definitive completion times of jobs, we manage to prove the competitive ratio of our algorithm with admittedly sophisticated analysis.

\subsection{Related Works}

The problem of minimizing the total weighted flow-time has been extensively studied in the online scenario. 
For the preemptive problem, Chekuri et al.~\cite{ChekuriKhanna01:Algorithms-for-minimizing} presented a \emph{semi-online} $O(\log^2 P)$-competitive algorithm for a single machine, where $P$ is the ratio of the largest to the smallest processing time of the instance.
Later, Bansal and Dhamdhere~\cite{Bansal:2007} proposed a $O(\log W)$-competitive algorithm, where $W$ is the ratio between the maximum and the minimum weights of the jobs. A lower bound of $\Omega(\min((\log W/\log \log W)^{\frac{1}{2}}, (\log \log P/ \log \log \log P)^{\frac{1}{2}})$ was shown in~\cite{Bansal:2009}. 
In contrast to the single-machine case, Chekuri et al.~\cite{ChekuriKhanna01:Algorithms-for-minimizing} showed a $\Omega(\min(\sqrt{P}, \sqrt{W}, {\frac{n}{m}}^\frac{1}{4}))$ lower bound for $m$ identical machines. 
For the online non-preemptive problem of minimizing the total weighted flow-time, Chekuri et al.~\cite{ChekuriKhanna01:Algorithms-for-minimizing} showed that any algorithm has at least $\Omega(n)$ competitive ratio for single machine where $n$ is the number of jobs.

In resource augmentation model, Anand et al.~\cite{AnandGarg12:Resource-augmentation} presented a scalable competitive algorithm for the preemptive problem on a set of unrelated machines. 
For the non-preemptive setting, Phillips et al.~\cite{PhillipsStein02:Optimal-time-critical} gave a constant competitive algorithm in identical machine setting that uses $m \log P$ machines (recall that the adversary uses $m$ machines).
They also showed that there exists a $O(\log n)$-machine $O(1)$-speed algorithm that returns the optimal schedule for the unweighted flow-time objective.
Epstein et al.~\cite{EpsteinVanStee06} proposed an $\ell$-machines $O(\min\{\sqrt[\ell]{P},\sqrt[\ell]{n}\})$-competitive algorithm for the unweighted case on a single machine. This algorithm is optimal up to a constant factor for constant $\ell$.

Lucarelli et al.~\cite{LucarelliThang16:Online-Non-preemptive} presented a strong lower bound on the competitiveness for the weighted flow-time problem on a single machine that uses arbitrarily faster machine than that of the adversary.
Choudhury et al.~\cite{ChoudhuryDas15:Rejecting-jobs} extended the resource augmentation model to allow \emph{rejection}, according to which the algorithm does not need to complete all jobs and some of them can be rejected.
Using a combination of speed augmentation and rejection, Lucarelli et al.~\cite{LucarelliThang16:Online-Non-preemptive} gave a constant competitive algorithm for the weighted flow-time problem on a set of unrelated machine. In particular, they showed that there exists a $O(1/\epsilon^2)$-competitive algorithm that uses machines with speed $(1+\epsilon)$ and rejects at most an $\epsilon$-fraction of jobs for arbitrarily small $\epsilon > 0$. 
Recently, Lucarelli et al. \cite{LucarelliThang16:Unweighted-Non-preemptive} provided a scalable competitive algorithm for the case of (unweighted) flow time where there is no speed augmentation.



\section{Definitions and Notations}
\subsection{Problem definition}

We are given a set $\mathcal{M}$ of unrelated machines and a set of jobs $\mathcal{J}$  that arrive online. Each job $j$ is characterized by its release time $r_j$ and its weight $w_j$. If job $j$ is executed on machine $i$, it has a processing requirement of $p_{ij}$ time units. The goal is to schedule jobs \emph{non-preemptively}. Given a schedule $\mathcal{S}$, the \emph{completion time} of the job $j$ is denoted by $C_j^{\mathcal{S}}$. The \emph{flow-time} of $j$ is defined as $F_j^{\mathcal{S}} = C_j^{\mathcal{S}} -r_j$, which is the total amount of time job $j$ remains in the system. The objective is to minimize the weighted flow-times of all jobs, i.e., $\sum_{j \in \mathcal{J}} w_j F_j^{\mathcal{S}}$. In the following section we formulate this problem as a linear program.

\subsection{Linear Programming Formulation}

The LP formulation presented below is an extension of those used in the prior works of \cite{AnandGarg12:Resource-augmentation, LucarelliThang16:Online-Non-preemptive}.  
For each job $j$, machine $i$ and time $t \geq r_j$, there is a binary variable $x_{ijt}$ which indicates if $j$ is processed or not on $i$ at time $t$. The problem of minimizing weighted flow-time can be expressed as:

$$\min \sum \limits_{i,j,t}  w_j \left(\frac{t - r_j}{p_{ij}} + 21 \right) x_{ijt} $$ 
\vspace{-0.15cm}
\begin{align}
&&\sum \limits_{i,t} \frac{x_{ijt}}{p_{ij}} &= 1  \hspace{2cm} \forall j  \label{runs}\\   
&&\sum_j x_{ijt} &\leq 1	  \hspace{2cm}  \forall i, t  \\
&& x_{ijt} &\in \{0,1\} \hspace{1cm} \forall i, j, t \geq r_j
\end{align}

The objective value of the above integer program is at most a constant factor than that of the optimal preemptive schedule. 
The above integer program can be relaxed to a linear program by replacing the integrality constraints of $x_{ijt}$ with $0 \leq x_{ijt} \leq 1$. The dual of the relaxed linear program can be expressed as follows: 
$$\max \sum_{j} \alpha_j - \sum \limits_{i,t} \beta_{it}$$
\vspace{-0.15cm}
\begin{align}
&&\frac{\alpha_j}{p_{ij}} - \beta_{it} &\leq w_j \left(\frac{t - r_j}{p_{ij}} + 21 \right) \hspace{1cm} \forall i, j, t \geq r_j \label{dual-const}\\
&&\beta_{it} &\geq 0 \hspace{1cm} \forall i, t
\end{align}

For the \emph{rejection model} considered in this work, it is assumed that the algorithm is allowed to reject jobs. Rejection can be interpreted in the primal LP by only considering constraints corresponding to non-rejected jobs. That is, the algorithm does not have to satisfy the constraint~(\ref{runs}) for rejected jobs. 

\subsection{Notations} \label{notation}

In this section, we define notations that will be helpful during the design and analysis of the algorithm. 
\begin{itemize}
	\item $t^-$ denotes the time just before $t$ that is, $t^- = t - \epsilon'$ for  an arbitrarily small value of $\epsilon' > 0$. 
	\item $U_i(t)$ denotes the set of pending jobs at time $t$ on machine $i$, i.e., 
		the set of jobs dispatched to $i$ that have not yet completed and also have not been rejected until $t$. 
	\item $\kappa_i(t)$ denotes the job currently executing on machine $i$ at time $t$. 
	\item $V_i(t)$ denotes the set of unprocessed jobs in $U_i(t)$ that is $V_{i}(t) = U_{i}(t) \setminus \{\kappa_{i}(t)\}$. Throughout this paper, we assume that the jobs in $V_i(t)$ are indexed in non-increasing order of their densities that $\delta_{i1} \geq \delta_{i2}, \ldots, \geq \delta_{i|V_i(t)|}$.
	\item $\nu_i(t)$ denotes the smallest density job in $V_i(t)$.
	\item $R^1_{i}(a,b)$ denotes the set of jobs rejected due to the \emph{prempt} rule  (to be defined later) during time interval $(a,b]$.
		In particular, $R^1_{i}(t)$ is the set of job rejected at time $t$ due to the \emph{prempt rule}.
	\item Similarly, $R^2_{i}(a,b)$ denotes the set of jobs rejected due to the \emph{weight-gap} rule (also to be defined later) during time interval $(a,b]$. In particular, $R^2_{i}(t)$ is the set of job rejected at time $t$ due to the weight-gap rule.
	\item $q_{ij}(t)$ denotes the remaining processing time of $j$ at a time $t$ on machine $i$. 
	\item $\delta_{ij}$ is the density of a job $j$ on machine $i$ that is $\delta_{ij} = \frac{w_i}{p_{ij}}$. 
	\item $S_{j}$ denotes the starting of job $j$ on some machine $i$. If a job is rejected before it starts executing, set $S_j = \infty$.
\end{itemize}
By the previous definitions, it follows that $R^1_{i}(r_{j}), R^2_{i}(r_{j}) \subseteq U(r_j^-) \cup \{j\}$  and $U(r_j) = (U(r_j^-) \cup \{j\}) \setminus \{R^1_{i}(r_{j}) \cup R^2_{i}(r_{j})\}$.  
%

\section{The Algorithm} 

In this section, we describe our algorithm. Specifically, we explain how to take the following decisions: \emph{dispatching} that is decide the machine assignment of jobs;  \emph{scheduling} that is decide which jobs to process at each time; and \emph{rejection}.  The algorithm is denoted by $\mathcal{A}$. Let $0 < \epsilon < 1$ be an arbitrarily small constant. Note that the proposed algorithm rejects an $O(\epsilon)$-fraction of the total weight of jobs and dispatches each job to a machine upon its arrival. 

\subsection{Scheduling policy}  
At each time $t$ if the machine $i$ is \emph{idle} either due to the rejection of a job or due to the completion of a job, then the algorithm starts executing the job $j$ with the highest density among all the jobs in $U_i(t)$, \textit{i.e.} $j = \arg \max_{h \in U_i(t)} \delta_{ih}$. In case of ties, the algorithm selects the job with the earliest release time.

\subsection{Rejection policies} \label{rejection policies}
Our algorithm uses two different rules for rejecting jobs. The first rule called as the \emph{preempt rule},  bounds the total weight of ``high priority'' jobs that arrive during the execution of a job. 
The second rule called as the \emph{weight-gap} rule, helps the algorithm to balance the total amount of weight of low density jobs. The algorithm associates two counters, $\ct^1_j$ and $\ct^1_j$, with each job $j$ which are both initialized to $0$ at $r_j$. 

\begin{enumerate}
\item{\textbf{Preempt rule}}: 
Let $j = \kappa_{i}(t)$ be the job processing on $i$ at time $t$.  During the processing of $j$, if a new job $j'$ is dispatched to $i$ then $\ct^{1}_{j}$ is incremented by $w_{j'}$. Let $k$ be the earliest job released and dispatched to machine $i$ during the execution of $j$ such that $\ct^{1}_{j} \geq w_j/\epsilon$, if it exists. At $r_{k}$, the algorithm interrupts the processing of $j$ and rejects it, that is $R^1_{i}(r_k) = \{j\}$. If no job is rejected due to the \emph{preempt} rule at $r_k$, then we set $R^1_{i}(r_{k}) = \emptyset$. 
\smallskip

\item{\textbf{Weight-gap rule}}: 
We associate a function $W_i(t): \mathbb{R}^{+} \rightarrow \mathbb{R}^{+}$  with each machine $i$ which is initialized to $0$ for every $t$. Informally $W_i(t)$ represents the total budget for future rejections.  If a job $j$ is dispatched to machine $i$ then $W_i(t)$ for $t \geq r_{j}$ is updated according to the following policy.

Let $V = V_i(r_{j}^-) \cup \{j\}$. Assume that the jobs in $V$ are indexed in  non-increasing order of their densities that is, 
$\delta_{i1} \geq \delta_{i2} \geq \ldots \geq \delta_{i\nu}$, where the job with index $\nu$ is the smallest density job in $V$. 
Note that the job $j$ is included in this ordering. 
Let $s$ be the smallest index in $\{1,2,\ldots, \nu\}$ such that:
\begin{align}	 \label{ub_wg}
&&\sum \limits_{h = s}^{\nu} w_h  \leq \epsilon(W_i(t^-) + w_j) < \sum \limits_{h = (s-1)}^{\nu} w_h
\end{align}
We say that no such job with index $s$ exists if and only if $w_{\nu} > \epsilon (W_i(t^-) + w_j)$. Algorithm~\ref{Weight-gap Rejection Rule} defines the set of jobs $R^2_i(r_j)$. The algorithm rejects the jobs in $R^{2}_{i}(r_{j})$ and updates $W_i(t)$ as follows:
\begin{align}
&&W_i(t) = \max\{0,W_i(r_{j}^-) + w_j - \sum \limits_{h \in R^{2}_{i}(r_{j})} w_h/\epsilon\}, \hspace{1cm} \forall t \geq r_j	\label{weight-update-equation}
\end{align}
\end{enumerate}
The following lemma describes some properties arising due to the \emph{weight-gap} rule.

\begin{lemma} \label{wg-rule-prop}
The following properties hold. 
	\begin{description}
		\item[(Property 1)] 
			If $R^2_i(r_j) = \{ \nu_i(r_j^-), j\}$ or $R^{2}_{i}(r_{j}) = \{(s-1),\ldots,v\}$ then $W_i(r_j) = 0$. 
		%
		%
		\item[(Property 2)] 
		$\epsilon W_i(t) < w_{\nu_i(t)}$ for every pair of $i,t$.
		\item[(Property 3)]  
		Let $w_{|R^2_i(r_j)|}$ denote the weight of smallest density job in $R^2_i(r_j)$. 
			If $j \notin R^2_i(r_j)$ then $\sum_{h \in R^2_{i}(r_{j})} w_h - w_{|R^2_i(r_j)|} \leq 2\epsilon w_j$.
		%
		\item[(Property 4)]  
			If $j \in R^2_{i}(r_{j})$, then  $R^2_{i}(r_{j}) =  \{j\} \text{~or~} \{j, \nu_{i}(r^{-}_{j})\}$.
	\end{description}
\end{lemma}
\begin{proof} 

Assume that there exists no job with index $s$. This implies that $w_{\nu} > \epsilon (W_i(t^-) + w_j)$. From Line~\ref{reject-small-wgts} in Algorithm~\ref{Weight-gap Rejection Rule}, it follows that $R^2_i(r_j) = \{j, \nu_i(r_j^-)\}$. Combining this with Equation~\ref{weight-update-equation}, it follows that $W_i(r_j) = 0$. Now, assume that there exists a job with index $s$. Then combining the fact that $R^2_i(r_j) = \{(s-1), \ldots, \nu\}$ with Inequality~\ref{ub_wg}, it follows that $W_i(r_j) = 0$. This completes the proof of Property 1.

Property 2 holds trivially at time $0$. Assume that it holds before the arrival of job $j$. We show that it holds after $t = r_j$. Consider the case $t < r_j$ since $W_i(t)$ is unchanged, the property holds. We now consider the case for $t \geq r_j$.  Suppose there exists no job with index $s$ and $R^2_i(r_j) = \emptyset$. From Inequality~\ref{ub_wg}, the property holds. If $R^2_i(r_j) \not = \emptyset$ then $R^2_i(r_j) = \{j, \nu_i(r_j^-)\}$ from Line~\ref{reject-small-wgts} in Algorithm~\ref{Weight-gap Rejection Rule}. From Property 1, it follows that $W_i(r_j) = 0$ and the property follows. Suppose there exists a job with index $s$ then $R^2_i(r_j) = \{(s-1), \ldots, \nu\}$ or $R^2_i(r_j) = \{s, \ldots, \nu\}$. In the former case, from Property 1 it holds that $W_i(r_j) = 0$ and the proof follows. In latter case that is $R^2_i(r_j) = \{s, \ldots, \nu\}$, combining Inequality~\ref{ub_wg} with Equation~\ref{weight-update-equation}, the proof follows for Property 2.  

Since $j \notin R^2_i(r_j)$ from Algorithm~\ref{Weight-gap Rejection Rule} it follows that $R^2_i(r_j) =  \{(s-1), \ldots, \nu\}$ or $R^2_i(r_j) = \{s, \ldots, \nu\}$. Consider the time $r_j^-$ just before the arrival of $j$. From Property 2, it follows that $\epsilon W_i(t^-) < w_{\nu_i(t^-)}$. From Equation~\ref{weight-update-equation}, it follows that $0 \leq W_i(t) \leq W_i(t^-) + w_j - \sum \limits_{h \in R^2_i(r_j) \setminus {(s-1)}} w_h/\epsilon$. Combining both facts, we have $$0 \leq w_{\nu_i(t^-)}/\epsilon + w_j - \sum \limits_{h \in R^2_i(r_j) \setminus {(s-1)} } w_h/\epsilon$$ 
Since $j$ is not rejected at $r_j$, the job $\nu_i(r_j^-) \in R^2_i(r_j)$ and  we get 
$$\sum \limits_{h \in R^2_i(r_j) \setminus \{(s-1) \cup \nu_i(r_j^-)\} } w_h/\epsilon \leq \epsilon w_j$$
If the job with index $(s-1)$ is not in $R^2_i(r_j)$, the property follows. On the other hand if the job with index $(s-1)$ is in $R^2_i(r_j)$ then from Line~\ref{s-1-small} in Algorithm~\ref{Weight-gap Rejection Rule}, it follows that $w_{(s-1)} \leq \epsilon w_j$. Combining two previous facts, Property 3 holds.

Assume that there exists no job with index $s$ then Property 4 holds trivially. Now, suppose that there is a job with index $s$. If $j$ is not the smallest density job in $V$, then $j \not \in R^2_i(r_j)$. Assume the contrary that is, $j$ is the not smallest density job in $V$ and $j \in R^2_i(r_j)$. From Property 2, we have that $\epsilon W_i(t^-) < w_{\nu_i(t^-)}$. This can rewritten as $\epsilon (W_i(t^-) + w_j) < w_{\nu_i(t^-)} + w_j$. Since $j$ is not the smallest density job in $V$, from Inequality~\ref{ub_wg} it follows that the index of $j$ is at most $(s-1)$. And If $j$ is in fact the job with index $(s-1)$. Then $j$ never rejected since $w_j < w_{(s-1)}/\epsilon$. This contradicts the assumption that $j \in R^2_i(r_j)$.  Hence if $j \in R^2_i(r_j)$ then $j$ is the smallest density job in $V$. From Property 2, we have that $\epsilon W_i(t^-) < w_{\nu_i(t^-)}$. Since $j$ is the smallest density job in $V$, $\nu_i(t^-)$ is the job with index $\nu-1$. Hence it follows that $\epsilon (W_i(t^-) + w_j) \leq w_j + w_{\nu_i(t^-)} = w_{\nu} + w_{(\nu-1)}$. The job with index $\nu-1$ correspond to the job with index $(s-1)$ in Line~\ref{s-1-rejected}. Hence, Property 4 holds. 
\end{proof}
\begin{algorithm}[H]
\begin{algorithmic}[1]
\IF {no job with index $s$ exists}
	\IF {$j$ is the not smallest density job in $V$}
    		\STATE No job is rejected that is, $R^2_i(r_j) := \emptyset$
    	\ELSE   \label{nr-j-is-last}
		\STATE \COMMENT {$j$ is the smallest density job in $V$ that is, $j$ is the job with index $\nu$}
		\STATE \COMMENT {$v_i(r_j^-)$ is the job with index $\nu-1$}	\label{v_1eqv-1}
		\IF {$p_{ij} \geq \epsilon p_{i(\nu-1)} $}
			\STATE No job is rejected that is, $R^2_i(r_j) := \emptyset$
		\ELSE 	\label{j-very-small}
			\STATE $\ct^2_{(\nu-1)} := \ct^2_{(\nu-1)} + w_j$
			\IF {$\ct^2_{(\nu-1)} \geq w_{(\nu-1)}$}
				\STATE Reject $j$ and $\nu_i(r_j^-)$ that is, $R^2_i(r_j) := \{j, \nu_i(r_j^-)\}$ \label{reject-small-wgts}
			\ELSE
				\STATE No job is rejected that is, $R^2_i(r_j) := \emptyset$  \label{no-small-wgts-reject}
			\ENDIF
		\ENDIF
	\ENDIF			
\ELSE
	\STATE \COMMENT{a job with index $s$ exists}
	\IF {$w_j \geq w_{(s-1)}/\epsilon$}	
		\STATE Reject jobs with indices $s-1, \ldots ,\nu$ in $V$ that is, $R^2_i(r_j) := \{s-1, \dots, \nu\}$ \label{s-1-small}
	\ELSE
	\STATE \COMMENT{$w_j < w_{(s-1)}/\epsilon$}
		\IF {$j$ is not one of the jobs in $s, \ldots, \nu$ that is, $j \notin \{s, \ldots, \nu\}$}
			\STATE Reject jobs with indices $s, \ldots, \nu$ in $V$ that is, $R^2_i(r_j) := \{s, \ldots, \nu\}$ \label{s-1-not-rejected}
		\ELSE
			\STATE \COMMENT	{$j \in \{s, \ldots, \nu\}$} \label{j-is-reject}
			\STATE  $\ct^2_{(s-1)} := \ct^2_{(s-1)} + w_j$
			\IF {$\ct^2_{(s-1)} \geq w_{(s-1)}$}
				\STATE Reject jobs with indices $s-1, \ldots ,\nu$ in $V$ that is, $R^2_i(r_j) := \{s-1, \dots, \nu\}$.  \label{s-1-rejected} 
			\ELSE
				\STATE  Reject jobs with indices $s, \ldots, \nu$ in $V$ that is, $R^2_i(r_j) := \{s, \ldots, \nu\}$.  \label{only-j-rejected}
			\ENDIF
		\ENDIF
	\ENDIF
\ENDIF
\end{algorithmic}
\caption{Weight-gap Rejection Rule}
\label{Weight-gap Rejection Rule}
\end{algorithm}

\begin{lemma} \label{frac-reject-preempt-rule}
The total weight of jobs rejected by the \emph{preempt rule} is at most $O(\epsilon)$-fraction of the total weight of jobs in $\mathcal{J}$. 
\end{lemma}
\begin{proof}
From \emph{preempt} rule, it follows that each job $j$ can be associated with a set of jobs such that their total weight is at most $w_j/\epsilon$. For every pair of $j, j'$ and $j \not=j'$, the intersection of the associated sets is empty and hence the lemma follows. 
\end{proof}

\begin{lemma} \label{frac-rejected-wgtgap-rule} 
The total weight of jobs rejected by the \emph{weight-gap} rule is at most $O(\epsilon)$-fraction of the total weight jobs in $\mathcal{J}$.
\end{lemma}
\begin{proof}
At the arrival of a job $j$ on machine $i$, $W_i(t)$ is incremented by $w_j$. Combining Inequality~\ref{ub_wg} with Property 3 in Lemma~\ref{wg-rule-prop}, it follows that the total weight of jobs rejected  with respect to $W_i$, is at most $2\epsilon$-fraction of total weight of jobs dispatched to $i$.

At the arrival of $j$, $\ct^2_{j'}$ for some $j'$ on $i$ may be incremented by $w_{j'}$. Assume that $j'$ is rejected due to the fact $\ct^2_{j'} \geq w_{j'}/\epsilon$. This can happen either at Line~\ref{reject-small-wgts} or at Line~\ref{s-1-rejected} in Algorithm~\ref{Weight-gap Rejection Rule}. We first focus at Line~\ref{s-1-rejected} that is, there exists a job with index $s$. In this case, $j'$ is rejected only after arrival of $w_{j'}/\epsilon$. Combining this fact with total weight of jobs rejected with respect to $W_i$, we have that total weight of jobs rejected is at most $3 \epsilon$. Now consider the case in Line~\ref{reject-small-wgts}. Observe that $\nu-1$ is the job with index $j'$. Here both $j'$ and $j$ is rejected. Since $p_{ij} < \epsilon p_{i(\nu-1)}$ and $\delta_{ij} \geq \delta_{i(\nu-1)}$, it follows that $w_{j} \leq \epsilon w_{(\nu-1)} = \epsilon w_{j'}$. In this case, the total weight of jobs rejected is at most $w_{j} + w_{j'} = (1+\epsilon)w_{j'} \leq  2w_{j'}$ and total weight of jobs accounted in $\ct^2_{j'}$ is $w_{j'}/\epsilon$. Combining this fact with total weight of jobs rejected with respect to $W_i$, we have that total weight of jobs rejected is $4\epsilon$. Hence the total weight of jobs rejected due to \emph{weight-gap} rule is $O(\epsilon)$.   
\end{proof}

\subsection{Dispatching policy}

When a new job $j$ arrives,  a variable $\Delta_{ij}$ is set. Intuitively, $\Delta_{ij}$ is the approximate increase in the total weighted flow-time objective if the job $j$ is assigned to the machine $i$ and $j$ is not rejected. Then, $\Delta_{ij}$ is defined as follows.

\begin{align*}
\Delta_{ij}  &=  
		w_j \sum \limits_{h \in V_i(r_j) : \delta_{ih} \geq \delta_{ij}} p_{ih} + p_{ij} \sum \limits_{h \in V_i(r_j): \delta_{ih} < \delta_{ij} } w_{h} \\
 		&+ w_j q_{i\kappa_i(r^{-}_j)} (r_j) \cdot \one_{\text{\{$\kappa_i(r^{-}_j)$ is not rejected currently due to preempt rule\}}} \\
		&- q_{i\kappa_i(r_j^-)}(r_{j}) \cdot \sum \limits_{h \in U_i(r_j) \setminus \{j\}} w_h \cdot \one_{\text{\small\{$\kappa_i(r^{-}_j)$  is currently rejected due to the preempt rule\}}}
\end{align*}
%
The first term corresponds to the flow-time of the new job $j$ due to waiting on jobs with higher density than $\delta_{ij}$ in $V_i(r_j)$. 
The second term corresponds to the delay of the jobs in $V_i(r_j)$ with smaller density than $\delta_{ij}$. The third and the fourth terms give corrections depending on whether job $\kappa_i(r_j^-)$ is rejected due to the \emph{preempt} rule. 

We now describe the dispatching policy of jobs to machines. 
At the arrival time of a job $j$, we hypothetically assign $j$ to every machine $i$ and compute the variables $\alpha_{ij}$. Finally, we assign $j$ to the machine that minimizes $\alpha_{ij}$. For notional purposes, we put an additional apostrophe to previously defined variables.  The additional apostrophe stands for the fact that these variables correspond to the case where we hypothetically assign $j$ to $i$. For example, $R^{2'}_{i}(r_{j})$ denote the set of rejected jobs due to the \emph{weight-gap rule} when $j$ is hypothetically assigned to $i$. Similarly, $W'_i(r_j)$ denote the function $W_i$ at $r_j$ in the case if $j$ is assigned to $i$.
Further, let $\rho = \rho_{ij}$ be an index of a job in $V_i(r_j^-)$ such that the following two inequalities hold simultaneously:
\begin{align*}
&&\sum \limits_{h = \rho}^{|V_i(r_j^-)|} w_h \leq W'_i(r_j) < \sum \limits_{h = (\rho-1)}^{|V_i(r_j^-)|} w_h 
\end{align*}
The variable $\alpha_{ij}$ is computed for each machine $i$ as follows:

\begin{align*}
\alpha_{ij} &=  \frac{20 w_j p_{ij}}{\epsilon} + w_j \sum \limits_{h \in V_i(r_j^-) : \delta_{ih} \geq \delta_{ij}} p_{ih} + w_j p_{ij} + p_{ij}\sum \limits_{h \in V_i(r_j^-) : \delta_{ij} > \delta_{ih}} w_{ih} - n_{ij} \\
\end{align*}
where $n_{ij}$ is defined as follows. 
\begin{align*}
n_{ij} &= 
	w_j \left(\sum \limits_{h \in V_i(r_j^-) :  \delta_{i\rho} \geq \delta_{ih} }  p_h + 
	\biggl( W'_i(r_j) -  \sum \limits_{h \in V_i(r_j^-) :  \delta_{i\rho} \geq \delta_{ih}} w_h \biggr) \frac{p_{i,(\rho-1)}}{w_{(\rho-1)}} \right) \\
		&  \hspace{10cm} \text{if } R^{2'}_{i}(r_{j}) = \{j\},  \\
n_{ij} &=	w_j \sum \limits_{h \in R'^2_{i}(r_{j})} p_{ih} \hspace{6.2cm}  \text{if } R'^2_{i}(r_{j}) = \{j, \nu_i(r_j^-)\}, \\
n_{ij} &=	p_{ij} \sum \limits_{h \in R'^2_{i}(r_{j})} w_h +  \epsilon^2 W'_i(r_j)  p_{ij}  \hspace{6cm}  \text{otherwise.}
\end{align*}
The algorithm assigns $j$ to machine $i^* = \arg \min_{i \in \mathcal{M}} \alpha_{ij}$. 
 
%

\subsection{Dual variables} \label{dual variables}
 

Suppose job $j$ is assigned to machine $i$. Assume $L_j$ represents the last time $t$ such that $j$ is in $U_i(t)$. Informally, $L_j$ is the time at which $j$ is removed from the queue of the machine $i$. Note that $j$ can be removed from $U_i(t)$ due to three following reasons:
\begin{enumerate}
 	\item If $j$ has being scheduled for $p_{ij}$ time units on machine $i$ then $L_j = C_j$ \label{tau_completion}
 	\item If $j$ is rejected due to \emph{preempt} rule \label{tau_preempt}
 	\item If $j$ is rejected due to \emph{weight-gap} rule.  \label{tau_weight-gap}
 \end{enumerate}
 	
In cases ~\ref{tau_preempt} and~\ref{tau_weight-gap}, $j$ is rejected due to the arrival of some job,  denoted by $\rej(j)$. Recall that $R^{1}_{i}(r_{j},L_{j})$ is the set of jobs that are rejected due to \emph{preempt rule} during the interval $(r_j, L_j]$ on machine $i$. Note that those jobs cause a decrease in the flow of $j$. Observe that $R^{1}_{i}(r_{j},L_{j})$ contains $j$ if $j$ is rejected due to the \emph{preempt} rule. We define the \emph{definitive completion} time, denoted by  $\widetilde{C}_{j}$, of a job $j$ as follows.

\begin{enumerate}
	\item If $j$ is not rejected due to the \emph{weight-gap} rule (corresponds to cases~\ref{tau_completion} and~\ref{tau_preempt}) .
		\begin{align}
			&&\widetilde{C}_{j} = L_j + \sum \limits_{h \in R^{1}_{i}(r_{j},L_{j})} q_{ih}(r_{\mathsf{rej}(h)}) \label{completion-time-j-not-rejected}
		\end{align}
	\item If $j$ is rejected due to the \emph{weight-gap} rule on the arrival of some job other than $j$ that is, $r_{j'}$ where $ j' \not= j$ 		

\begin{align}
			&&\widetilde{C}_{j} = L_j + \sum \limits_{h \in R^{1}_{i}(r_{j},L_{j})} q_{ih}(r_{\mathsf{rej}(h)}) 
		  	+ \sum \limits_{h \in U_i(L_{j}) : \delta_{ih} \geq \delta_{ij}}q_{ih}(L_{j}) + 
		  	\sum \limits_{h \in R^2_{i}(r_{\rej(j)}): \delta_{ih} \geq \delta_{ij}} p_{ih} \label{completion-time-j-rejected-by-k}
		\end{align}
	\item If $j$ is immediately rejected (i.e., $j \in R^{2}_{i}(r_{j})$) and job $\nu_{i}(r^{-}_{j})$ is also rejected due to the arrival of $j$.
		\begin{align}
			&&\widetilde{C}_{j} = L_j + p_{ij} + \sum \limits_{h \in U_i(L_{j})} q_{ih}(L_{j}) \label{completion-time-j-nu-rejected}
		\end{align}
	\item If $j$ is immediately rejected and it is the only job rejected due to the \emph{weight gap} rule at $r_j$. Denote $\rho = \rho_{ij}$. 
		\begin{align}
				 \widetilde{C}_{j} &= L_j + p_{ij} + \sum \limits_{h \in V_i(L_{j}^-) : \delta_{ih} > \delta_{i(\rho-1})} p_{ih} +\left(1 - \frac{W_i(L_{j}) -  \sum \limits_{h \in V_i(L_{j}^-):  \delta_{i\rho} \geq \delta_{ih}} w_h}{w_{(\rho-1)}}\right) p_{i(\rho-1)} \nonumber \\
				 &+ q_{i\kappa_i(L_j)}(L_j). \one_{\{R^1_i(L_j) = \emptyset\}}
		\end{align}
\end{enumerate}
This completes the description of the \emph{definitive completion time}.
\bigskip

Let $Q_i(t)$ denote the set of jobs that have not been definitely completed that is 
$$Q_{i}(t) := \{j: j \text{ has been assigned to } i,  t < \widetilde{C}_{j}\}.$$
Next, we define the notion of \emph{artificial fractional weight} of a job $j \in Q_{i}(t)$,
\begin{align*}
w^{f}_j(t) &= 
	\begin{cases}
		w_j &\text{~if~} r_j \leq t \leq \widetilde{C}_{j} - p_{ij}  \\
		 w_j \left( \frac{\widetilde{C}_{j}-t}{p_{ij}} \right)	 &\text{~if~} \widetilde{C}_{j} - p_{ij} < t <\widetilde{C}_{j} 
	\end{cases}
\end{align*}
Now, we have all the necessary tools to set dual variables. At the arrival of job $j$,  set 
$$
\alpha_j = \left(\frac{\epsilon}{1+\epsilon}\right) \min_{i \in \mathcal{M}} \alpha_{ij}
$$ 
and  never change this value again. The second dual variable $\beta_{it}$ is set to 
$$
\frac{\epsilon}{(1+\epsilon) (1+\epsilon^2)} \sum \limits_{h \in Q_i(t)} w^{f}_h(t)
$$ 

Let $Q^{R}_i(t) \subseteq Q_i(t)$ be the set of jobs that are rejected due to the \emph{weight-gap rule} and are not yet definitively completed until time $t$. 

\begin{lemma} \label{monotone}
For fixed time $t$, $\beta_{it}$ may only increase as new jobs arrive and some old jobs might get rejected. 
\end{lemma}
Observe that above lemma holds as jobs are removed from $Q_i(t)$ only after their \emph{definitive completion time}. Thus a job that might have already completed its execution on a machine or rejected, can still be present in the $Q_i(t)$. During the analysis, we will show that the dual constraint corresponding to job $j$ are feasible at $r_j$. Since $\beta_{it}$ only increases with respect to the arrival of new jobs, the feasibility holds for all $t \geq r_j$. 
\section{Analysis}

We present first two technical lemmas which are important for the analysis of our primal-dual algorithm. In Lemma~\ref{lem:main-inequality}, we relate the weight of rejected jobs in $Q^{R}_i(t)$ to the weight of jobs pending in $U_i(t)$. This will help us in proving the feasibility of dual constraints in Lemma~\ref{feasible-dual-constraints-j-not-rejected}, Lemma~\ref{feasible-dual-constraints-j-nu-rejected} and Lemma~\ref{feasible-dual-constraints-j-only-rejected}. In Lemma~\ref{lem-part-dual-obj}, we show that the negative parts in the definition of  $\alpha_j$s' are relatively small. This will help us to bound the value of the dual objective. 
\medskip 

\begin{lemma} \label{lem:main-inequality}
Let $\kappa = \kappa_{i}(t)$. For any machine $i$ and any time $t$, it holds that  
$\frac{w_\kappa}{p_{i\kappa}} q_{i\kappa}(t) + \sum \limits_{h \in V_i(t)} w^f_h(t) - W_i(t) \leq \frac{1}{\epsilon} \sum \limits_{h \in Q^R_i(t)} w_h(t) $.
\end{lemma}

\begin{proof} 

We prove by induction on the arrival of jobs. The base case holds trivially when no job has been released.  Assume the inequality (for every time $t$) holds before the arrival of job $j$ on machine $i$, we show that it holds after the job $j$ arrives. 

Fix a machine $i$. For the rest of this proof, we drop the machine index $i$. For example $R^2(r_j) = R^2_i(r_j)$. Assume that jobs are indexed according to their release time. Let $W(j,t)$ denote the value of $W(t)$ at time $t$ when job $1,2, \ldots, j$ have been released. Let $V(j,t)$ be the set of unprocessed jobs at time $t$ when jobs $1, 2, \ldots, j$ have been released. Let $Q^R(j,t)$ be the set of rejected jobs which have not been definitely completed at time $t$ when jobs $1, 2, \ldots, j$ have been released. Note that $Q^R(j, r_j) = Q^R(j-1, r_j) \cup R^2(r_j)$.
We consider two cases whether $j \in R^2(r_j)$ or $j \notin R^2(r_j) $. 

\begin{enumerate}
\item \textbf{Job $j$ is not immediately rejected at $r_j$ that is, $j \notin R^2(r_j).$} 

It follows from the \emph{weight-gap} rule that $w_j \leq W(j,r_j) - W(j-1, r_{j}) + \frac{1}{\epsilon} \sum_{h \in R^2(r_j)}  w_h$. Note that this rule rejects small density jobs so all jobs in $R^2(r_j)$ have definitely completed time after the smallest density job in $V(j,t)$. Therefore, for all time $t \in [r_j, C_\ell)$ where $\ell$ is the last completed job in the current schedule it holds that 
\begin{align}	\label{eq:weight-ineq}
w_j \leq W(j, t) - W(j-1, t) + \frac{1}{\epsilon} \sum \limits_{h \in R^2(r_j)}  w_h
\end{align}

Consider an arbitrary time $t$. If $t > C_\ell$ then the lemma inequality trivially holds since the left hand side is non-positive while the right hand side is non-negative. In the remaining, we consider $t \in [r_j, C_\ell)$. The set $V(j, t)$ has at most one additional job $j$ (and possibly some jobs have been removed). In other words, $V(j, t) \subset V(j-1, t) \cup \{j\}$.
We have 
\begin{align*}
&\sum_{h \in V(j, t)} w^f_h(t) - W(j, t) \\
&= \sum_{h \in V(j-1, t)} w^f_h(t) - W(j-1, t) + \sum_{h \in V(j,t) \setminus V(j-1, t)} w^f_h(t)  - W(j, t) + W(j-1, t) \\
&\leq \frac{1}{\epsilon} \sum_{h\in Q^R(j-1, t)} w_h(t) + \frac{1}{\epsilon} \sum_{h \in R^2(r_j)}  w_h - \frac{w_\kappa}{p_{i\kappa}} q_{i\kappa}(t) \\
&= \frac{1}{\epsilon} \sum_{h \in Q^R(j, t)} w_h(t) - \frac{w_\kappa}{p_{\kappa}} q_{\kappa}(t)
\end{align*}
where the second inequality follows the induction hypothesis and the set $V(j,t) \setminus V(j-1, t)$ contains at most job $j$; the last inequality is due to Inequality (\ref{eq:weight-ineq}).
\medskip

\item \textbf{Job $j$ is rejected at $r_j$ that is, $j \in R^2(r_j).$} 

If $|R^2(r_j)|  = \{j, \nu_i(r_j^-\}$, then according to the Equation (\ref{completion-time-j-nu-rejected}) \emph{definitive completion time} of $j$ is later than the completion time of all job in $V(j, t)$. It follows that $Q^R(j, t) = Q^R(j-1, t) \cup R^2(r_j), \forall t \in  [r_j, C_\ell)$ where $\ell$ is the last completed job in the current schedule. Thus, the lemma holds as in the Case 1. 

Now, we consider the case when $|R^2(r_j)| = \{j\}$. Unlike the previous case, job $j$ may complete earlier than the last job $\ell$ in $V(j,t)$. Let $\rho = \rho_{ij}$. Suppose $t$ be an arbitrary time smaller than $t'$ where $t'$ is defined as follows:
$$
	t' = r_j + \sum \limits_{h \in V(r_j^-): \delta_{h} > \delta_{(\rho-1)}} p_{h} +\left(1 - \frac{W(j, r_j) -  \sum \limits_{h \in V(r_j^-):  \delta_{\rho} \geq \delta_{h}} w_h}{w_{(\rho-1)}}\right)   p_{(\rho-1)} + q_{\kappa_i(r_j)}(r_j). \one_{\{R^1(r_j) = \emptyset\}} \\
$$
Then it holds that $Q^R(j, t) = Q^R(j-1, t) \cup R^2(r_j)$. As in the above case, the lemma holds for all $t \leq t'$. 

Now we focus on the time $t > t'$. At time $t'$, all the jobs of density larger than $\rho-1$ have been completed. Thus, the total fractional weight in the queue at $t'$ is  
\begin{align*}
	w_{(\rho-1)} \left(\frac{W(j, t') -  \sum \limits_{h \in V(r_j^-):  \delta_{\rho} \geq \delta_{h}} w_h}{w_{(\rho-1)}}\right)  + \sum \limits_{h \in V(t')} w^f_h(t')  &=  W(j, t')
\end{align*}
Hence, the left hand side of the inequality is at most $0$ and the lemma follows.
\end{enumerate}
\end{proof}

\begin{lemma} \label{lem-part-dual-obj} 
Let $J_i(t)$ denote the set of jobs dispatched to machine $i$ until the time $t$ that is, $J(i) = \bigcup \limits_{t' \leq t} U_i(t')$. Then the following inequality holds at all time and for all $i \in \mathcal{M}$
\begin{align}
	&&\mathcal{D}^1 - \mathcal{D}^2  \leq \mathcal{B}^1 + \mathcal{B}^2 + \mathcal{B}^3 \label{weird-inequality}
\end{align}
where 
\begin{align*}
&\mathcal{D}^1 = \sum _{j \in J_i(t) \setminus  R^2_i(r_j)} \left( \epsilon^2 W_i(r_j) p_{ij} - w_j p_{i, \nu_i(r_j^-)}.\one_{\{j = \nu_i(r_j) \text{~and~} p_j < \epsilon p_{i\nu_i(r_j^-)} \}} \right), \\
&\mathcal{D}^2 =\sum \limits_{j \in R^2_i(r_j)} \left(w_j p_{i,\nu_i(r_j)}. \one_{\{|R^2_i(r_j)| = 1\}} + w_{\nu_i(r_j^-)} p_{i,\nu_i(r_j^-)}. \one_{\{|R^2_i(r_j)| > 1\}} \right),  \\
&\mathcal{B}^1 = \sum \limits_{j \in R^2_i(0,t)} w_j p_{ij},  \hspace{0.2cm} \mathcal{B}^2 = \sum \limits_{j \in J_i(t) \setminus \{R^2_i(0, t) \cup U_i(t)\}} w_j p_{ij}  + \epsilon W_i(t) p_{i,\nu_i(t)} \text{~and~} \\
&\mathcal{B}^3 = \sum \limits_{j \in J_i(t)} w_j p_{ij} /\epsilon.
\end{align*}
\end{lemma}
\begin{proof} 
Observe that  terms on the left side of the equation that is $\mathcal{D}^1$ and $\mathcal{D}^2$ change due to the arrival of a new job. Whereas on the right hand side of the equation, the first term $\mathcal{B}^1$ changes due to the dispatch of a job on machine $i$ whereas the second term $\mathcal{B}^2$ changes due to the arrival and the completion of a job on $i$. 

Consider the case when $j$ completes its processing on  $i$ at time $t$. If $j \not = \nu_i(t^-)$, then the second term $\mathcal{B}^2$ increases as a term of $w_jp_{ij}$  is added whereas there is no change in the $\mathcal{D}^1$ and $\mathcal{D}^2$. On other hand, if $j$ happens to be the smallest density job at $t^-$ then the change in $\mathcal{B}^2$ is positive that is, $w_j p_{ij} - \epsilon W(t) p_{ij} > 0$. This is due to the fact that the algorithm maintains the invariant that $\epsilon W(t) < w_{\nu_i(t)}$ for all $t$. Thus, we show in the following analysis that the above inequality holds when a new job $j$ arrives.

Fix a machine $i$. For the rest of this proof, we drop the machine index $i$. In the following let $\Delta_{\mathcal{B}}$ and $\Delta_{\mathcal{D}}$ represent the change on the right and the left hand side of Inequality (\ref{weird-inequality}), respectively. We split the proof into two separate cases depending upon if some jobs are rejected or not due to the \emph{weight-gap} rule at $t = r_j$.  
\bigskip 

\begin{enumerate}
\item \textbf{No job is rejected due to the weight-gap rule at $t$}. 

Then the total change on the right side can be written as:
\begin{align}
\Delta_{\mathcal{B}} &=  \epsilon(W(t) p_{\nu(t)} -  W(t^-) p_{\nu(t^-)}) + w_j p_{j}/\epsilon. \label{no-rejected}
\end{align}
\begin{enumerate}
\item \textbf{Job $j$ is not the smallest density job in $V(t)$}. 

Thus the job with the smallest density is same before and after the arrival of $j$ that is, $\nu(t) = \nu(t^-)$. From Equality (\ref{no-rejected}) we get: 
\begin{align*}
\Delta_{\mathcal{B}} &\geq \epsilon (W(t) p_{\nu(t)}  - W(t^-) p_{\nu(t)} ) \\
&= \epsilon w_j p_{\nu(t)} &\text{since $W(t) = W(t^-) + w_j$} \\
&\geq  \epsilon w_{\nu(t)} p_{j}   &\text{since $\delta_{j} \geq \delta_{\nu(t)}$} \\
&> \epsilon^2 W(t) p_{j}  = \Delta_{\mathcal{D}} & \text{since $w_{\nu(t)} > \epsilon W(t) $}
\end{align*}
\item \textbf{Job $j$ is the smallest density job in $V(t)$}. 

This implies that $\nu(t) = j$ and $\nu(t) \not = \nu(t^-)$. Then we have,
\begin{align}
\Delta_{\mathcal{B}} &=   w_j p_j/\epsilon + \epsilon(W(t) p_{j}  -  W(t^-) p_{\nu(t^-)}) \nonumber \\
&=   w_j p_j/\epsilon + \epsilon w_j p_j + \epsilon W(t^-) (p_j - p_{\nu(t^-)}) &\text{since $W(t) = W(t^-) + w_j$} \label{1.2}
\end{align}
\begin{enumerate}
\item \textbf{$p_j \geq \epsilon p_{\nu(t^-)}$}

Then from Inequality (\ref{1.2}) we have:
\begin{align*}
\Delta_{\mathcal{B}} &\geq   w_j p_j/\epsilon + \epsilon w_j p_j + \epsilon W(t^-) (p_j - p_j/\epsilon) &\text{since $p_j \geq \epsilon p_{\nu(t^-)}$}\\
&=  w_j p_j/\epsilon+ \epsilon w_j p_j - \epsilon W(t^-) (p_j/\epsilon - p_j) \\
&\geq  w_j p_j/\epsilon + \epsilon w_j p_j - \epsilon W(t) (p_j/\epsilon - p_j) &\text{since $W(t) = W(t^-) + w_j$ }\\
&>  w_j p_j/\epsilon + \epsilon w_j p_j - w_j (p_{j}/\epsilon - p_{j}) &\text{since $j = \nu(t)$ and $w_{\nu(t)} > \epsilon W(t)$}\\
&\geq  \epsilon w_j p_j + w_j p_j >  \epsilon^2 W(t) p_j = \Delta_{\mathcal{D}}
\end{align*}
\medskip
\item \textbf{$p_{j} < \epsilon p_{\nu(t^-)}$}

From Inequality (\ref{1.2}) we have:
\begin{align*}
\Delta_{\mathcal{B}} &\geq   w_j p_{j}/\epsilon + \epsilon w_j p_{j} - \epsilon W(t^-) p_{\nu(t^-)} \\
&\geq \epsilon^2 W(t) p_{j}  - \epsilon W(t^-) p_{\nu(t^-)} &\text{since $w_j > \epsilon W(t)$}\\
&\geq \epsilon^2 W(t) p_{j} - \epsilon W(t) p_{\nu(t^-)} &\text{since $W(t) = W(t^-) + w_j $} \\
&\geq \epsilon^2 W(t) p_{j} -  w_j p_{\nu(t^-)}  = \Delta_{\mathcal{D}}  &\text{since $w_j > \epsilon W(t) $}
\end{align*}
\end{enumerate}
\end{enumerate}
\bigskip

\item \textbf{Some jobs are rejected due to the weight-gap rule at $t = r_j$}. 

Using Equation (\ref{weird-inequality}), we have
\begin{align}
\Delta_{\mathcal{B}} =  w_j p_j/\epsilon + \sum \limits_{h \in R^2(r_j)} w_h p_{h} + \epsilon (W(t) p_{\nu(t)} -  W(t^-) p_{\nu(t^-)}) \label{d-rhs}
\end{align}
In next two sub-cases we assume that job $j$ is not immediately rejected at $t$ that is, $j \not \in R^2(t)$ and therefore in $\Delta_{\mathcal{D}}$ we have only term corresponding to $\mathcal{D}^1$. Whereas in the last two sub-cases, there is only term corresponding to $\mathcal{D}^2$. \\
\begin{enumerate}
\item \textbf{Job $j$ is not immediate rejected at $t = r_j$ and $W(t) >0$}. 

We have $j \not \in R^2(t)$ and $W(t) > 0$. This corresponds to the Line~\ref{s-1-not-rejected} in the \emph{weight-gap} rule. Then it follows that $\delta_{j} \geq \delta_{\nu(t)}$ and $\delta_{\nu(t)} \geq \delta_{h}, \forall h \in R^2(t)$.  Also note that the job $\nu(t^-)$ is rejected at $t$ that is, $\nu(t^-) \in R^2(t)$. Hence Equation (\ref{d-rhs}) can be rewritten as 
\begin{align*}
\Delta_{\mathcal{B}} &=  w_j p_j/\epsilon + \sum \limits_{h \in R^2(t)} w_h p_{h}+  \epsilon (W(t) p_{\nu(t)} -  W(t^-) p_{\nu(t^-)}) \\
&>  \sum \limits_{h \in  R^2(t)} w_h p_{h}+ \epsilon W(t) p_{\nu(t)} - w_{\nu(t^-)} p_{\nu(t^-)} \text{\hspace{1.5cm} since $w_{\nu(t^-)} > \epsilon W(t^-) $ }\\ 
&= \sum \limits_{h \in  R^2(t) \setminus \nu(t^-)}w_h p_{h}+  \epsilon W(t) p_{\nu(t)}   \text{\hspace{3.0cm} since $\nu(t^-) \in R^2(t)$}\\
&\geq \sum \limits_{h \in  R^2(t) \setminus \nu(t^-)} w_h p_{h} +  \epsilon \frac{W(t)}{w_j} w_{\nu(t)} p_{j} \text{\hspace{2.5cm} since $\delta_{j} \geq \delta_{\nu(t)}$} \\
&> \sum \limits_{h \in  R^2_i(t) \setminus \nu(t^-)} w_h p_{h} +\epsilon^2 W(t) p_{j} \text{\hspace{1cm} from Line~\ref{s-1-not-rejected} in Algorithm~\ref{Weight-gap Rejection Rule}, $w_j < w_{\nu(t)}/\epsilon$} \\
&\geq \Delta_{\mathcal{D}}
\end{align*}

\item \textbf{Job $j$ is not immediately rejected at $t = r_j$ and $W(t) = 0$}. 

Thus, we have $j \notin R^2(t)$ and $W(t) = 0$. This corresponds to Line~\ref{s-1-small} and Line~\ref{s-1-not-rejected} in the \emph{weight-gap} rule. Since $j$ is not rejected, it follows that $\delta_{j} \geq \delta_{h}, \forall h \in R^2(t)$. Using Equation (\ref{d-rhs}), we have that:
\begin{align*}
\Delta_{\mathcal{B}} &\geq  \sum \limits_{h \in R^2(r_j)} w_h p_{h}+ \epsilon (W(t) p_{\nu(t)} -  W(t^-) p_{\nu(t^-)} ) \\
&\geq  \sum \limits_{h \in R^2(r_j)} w_h p_{h} - \epsilon W(t^-) p_{\nu(t^-)} \text{\hspace{4.5cm} since $W(t) = 0$}\\
&\geq  \sum \limits_{h \in R^2(r_j)} w_h p_{h} - w_{\nu(t^-)} p_{\nu(t^-)}  \text{\hspace{4.5cm} since $\epsilon W(t^-) \leq w_{\nu(t^-)}$} \\
&\geq   \sum \limits_{h \in R^2(r_j) \setminus \nu(t^-)} w_h p_{h}\geq \Delta_{\mathcal{D}}
\end{align*}

\item \textbf{Job $j$ is immediately rejected at $t = r_j$ and $R^2(t) = \{j\}$}. 

It follows that $|R^2(r_j)| = 1$. This corresponds to Line~\ref{only-j-rejected} in the \emph{weight-gap} rule. Thus, we have $\nu(t) = \nu(t^-)$. Using Equation (\ref{d-rhs}), we have,
\begin{align*}
\Delta_{\mathcal{B}} &\geq  w_j p_{j} + \epsilon (W(t) p_{\nu(t)} -  W(t^-) p_{\nu(t^-)}) \\
 &\geq  w_j p_{j} +  \epsilon (W(t) p_{\nu(t)} -  W(t^-) p_{\nu(t)} ) &\text{since $\nu(t) = \nu(t^-)$} \\
 &\geq  w_j p_{j} +  \epsilon \left(W(t^-) + w_j - \frac{w_j}{\epsilon} \right) p_{\nu(t)} -  \epsilon W(t^-) p_{\nu(t)} &\text{since $W(t) = W(t^-) + w_j - w_j/\epsilon$}\\
&\geq   w_j p_{j} - (1 - \epsilon) w_j p_{\nu(t)} \\
&\geq  - w_j p_{\nu(t)} \geq \Delta_{\mathcal{D}}
\end{align*}
\item \textbf{Job $j$ is immediately rejected at $t = r_j$ and $R^2(t) = \{\nu(t^-), j\}$}. 

Thus we have that $R^2(r_j) > 1$. This corresponds to  cases~\ref{s-1-rejected} and \ref{reject-small-wgts} in the \emph{weight-gap} rule. Then it follows that the job $\nu(t^-)$ is rejected along with $j$. From Property 1 in Lemma~\ref{wg-rule-prop} it follows that $W(t) = 0$. Using  Equation (\ref{d-rhs}), we have,
\begin{align*}
\Delta_{\mathcal{B}} &\geq w_j p_{j} - \epsilon W(t^-) p_{\nu(t^-)}\\
 &\geq w_j p_{j} -  w_{\nu(t^-)} p_{\nu(t^-)} \geq \Delta_{\mathcal{D}} &\text{since $w_{\nu(t)} > \epsilon W(t^-)$} \\
\end{align*}
\end{enumerate}
\end{enumerate}

\end{proof}

\begin{corollary}    \label{coro-bound-partof-objective}
Let $J_i \subseteq \mathcal{J}$ be the set of jobs dispatched to machine $i$ that $J_i = \bigcup \limits_{t \geq 0} U_i(t)$. Then the following inequality holds for every machine $i \in \mathcal{M}$,
\begin{align*}
&&\sum _{j \in \mathcal{J}_i \setminus R^2_i(r_j)}  \epsilon^2 W_i(r_j) p_{ij} \leq \left(\frac{5}{\epsilon} \right) \sum_{j \in \mathcal{J}_i} w_j p_{ij}
\end{align*}
\end{corollary}

\begin{proof}
From Lemma~\ref{lem-part-dual-obj}, it immediately follows that
\begin{align*}
&\sum _{j \in J_i \setminus R^2_i(r_j)} (\epsilon ^2 W_i(r_j) p_{ij} - w_j p_{i,\nu_i(r_j^-)}. \one_{\{j = \nu(r_j) \text{~and~} p_j < \epsilon p_{i,\nu_i(r_j^-)}\}}) \\
- &\sum_{j \in R^2_i(r_j)} \left(w_j p_{i,\nu_i(r_j)}. \one_{\{|R^2_i(r_j)| = 1\}} + w_{\nu_i(r_j^-)} p_{i,\nu_i(r_j^-)}. \one_{\{|R^2_i(r_j)| > 1\}} \right) \leq \frac{2}{\epsilon} \left( \sum \limits_{j \in J_i } w_j p_{ij} \right)\\
\\
&\text{Rearranging the terms, we get} \\
&\epsilon^2 \sum _{j \in J_i \setminus R^2_i(r_j)} W_i(r_j) p_{ij} \\
- &\sum_{j \in R^2_i(r_j)} \left(w_j p_{i,\nu_i(r_j)}. \one_{\{|R^2_i(r_j)| = 1\}} + w_{\nu_i(r_j^-)} p_{i,\nu_i(r_j^-)}. \one_{\{|R^2_i(r_j)| > 1\}} \right) \\
&\leq \frac{2}{\epsilon} \left( \sum \limits_{j \in J_i } w_j p_{ij} \right) +\sum _{j \in J_i \setminus R^2_i(r_j)} w_j p_{i,v_i(r_j^-)}. \one_{\{j = v(r_j) \text{~and~} p_j < \epsilon p_{iv_i(r_j^-)}\}} \\
& \leq \frac{2}{\epsilon} \left( \sum \limits_{j \in J_i } w_j p_{ij} \right) + \sum _{h \in J_i} p_{ih} \sum_{j \in J_i: j = \nu(r_j), h = \nu(r_j^-)} w_j \\
&  \leq \frac{2}{\epsilon} \left( \sum \limits_{j \in J_i } w_j p_{ij} \right) + \sum _{h \in J_i} p_{ih} w_h/\epsilon  \\
& \text{since $\ct^2_{h} < w_h /\epsilon$, otherwise $h$ is rejected due to Line~\ref{reject-small-wgts} in Algorithm~\ref{Weight-gap Rejection Rule}} \\
&\leq \frac{3}{\epsilon} \left( \sum \limits_{j \in J_i } w_j p_{ij} \right) \\
&\text{Rearranging the terms again, we get} \\
&\epsilon^2 \sum _{j \in J_i \setminus R^2_i(r_j)} W_i(r_j) p_{ij} \\
&\leq \frac{3}{\epsilon} \left( \sum \limits_{j \in J_i } w_j p_j \right) + \sum_{j \in R^2_i(r_j)} \left(w_j p_{i,\nu_i(r_j)}. \one_{\{|R^2_i(r_j)| = 1\}} + w_{\nu_i(r_j^-)} p_{i,\nu_i(r_j^-)}. \one_{\{|R^2_i(r_j)| > 1\}} \right) \\
&\leq \frac{4}{\epsilon} \left( \sum \limits_{j \in J_i } w_j p_j \right) + \sum_{j \in R^2_i(r_j)} \left(w_j p_{i,\nu_i(r_j)}. \one_{\{|R^2_i(r_j)| = 1\}}\right)  \hspace{1.5cm} \text{since $R^2_i(r_j) = \{j, \nu_i(r_j^-)\}$} \\
&\leq \frac{4}{\epsilon} \left( \sum \limits_{j \in J_i } w_j p_j \right) + \sum_{h \in J_i} p_{ih} \sum_{j \in J_i: h = \nu_i(r_j^-) = \nu_i(r_j)} w_j \\ 
&\leq \frac{5}{\epsilon}  \left( \sum \limits_{j \in J_i } w_j p_j \right) 
\end{align*}
The last inequality holds since $\ct^2_{h} < w_h /\epsilon$, otherwise $h$ is rejected in Line~\ref{s-1-rejected} in Algorithm~\ref{Weight-gap Rejection Rule}.
Thus, the corollary follows. 
\end{proof}

Now we show the proof of dual feasibility for each job $j$ on every pair of $i,t$. Thus, for a given machine $i$, $j$ may or may not be assigned to $i$. by the algorithm.  
\begin{lemma} \label{feasible-dual-constraints-j-not-rejected}
	Suppose that a job $j$ is not immediately rejected at $r_j$ when $j$ is hypothetically assigned to $i$. Then, the dual constraint (\ref{dual-const}) corresponding to $j$ holds.
\end{lemma} 
\begin{proof}
Fix a machine $i$ and the time $t$. Property~\ref{monotone} states that for any fixed $t \geq r_j$, the value of $\beta_{it}$ may only increase. Hence it is sufficient to prove the above inequality at $r_j$, assuming that no job arrives after $r_j$. Let $R_i(r_j)$ denote the set of job rejected due to arrival of $j$, that is $R_i(r_j) = R^1_i(r_j) \cup R^2_i(r_j)$. From the definition of $\alpha_j$, we have:

\begin{align}
\frac{\alpha_j}{p_{ij}} \leq  \left(\frac{\epsilon}{1+\epsilon} \right) \frac{\alpha_{ij}}{p_{ij}}  &=  \frac{\epsilon}{1+\epsilon} \left(\frac{20 w_j}{\epsilon} + \frac{w_j}{p_{ij}} \sum \limits_{h \in V_i(r_j^-): \delta_{ih} \geq \delta_{ij}} p_{ih} + w_j + \sum \limits_{h \in V_i(r_j^-): \delta_{ij} > \delta_{ih}} w_{ih} \right)  \nonumber  \\
& - \frac{\epsilon}{1+\epsilon}  \left( \sum \limits_{h \in R^{2'}_i(r_j)} w_{ih} +  \epsilon^2 W'_i(r_j)\right) \label{feasible-1}
\end{align}

\begin{claim}
It holds that 
\begin{align}
\frac{\alpha_j}{p_{ij}} &\leq \frac{\epsilon}{1+\epsilon} \left(\frac{20  (1+\epsilon)w_j}{\epsilon} + \frac{w_j}{p_{ij}} \sum \limits_{h \in V_i(r_j): \delta_{ih} \geq \delta_{ij}} p_{ih} + \sum \limits_{h \in V_i(r_j): \delta_{ij} > \delta_{ih}} w_{ih} - \epsilon^2 W_i(r_j) \right) \label{feasible-A}
\end{align}
\end{claim}

\begin{proof}
Dispatching $j$ to machine $i$ changes the set of pending job in $U_i(t)$.  Accordingly, we have two sub-cases.
\medskip

\begin{enumerate}
\item \textbf{Job $j$ is dispatched to machine $i$}. The jobs $R^{2'}_i(r_j)$  are removed from the set $U_i(r_j^-)$ due to the \emph{weight-gap} rule. Furthermore, the job $\kappa_i(r_j^-)$ might be rejected due to the \emph{preempt} rule. From these two observations, it follows that $U_i(r_j) =  (U_i(r_j^-) \cup \{j\}) \setminus R_i
(r_j)$, Inequality~\ref{feasible-1} can be rewritten as: 
\begin{align*}
\frac{\alpha_j}{p_{ij}} &\leq \frac{\epsilon}{1+\epsilon} \left(\frac{20 w_j}{\epsilon} + \frac{w_j}{p_{ij}} \sum \limits_{h \in V_i(r_j): \delta_{ih} \geq \delta_{ij}} p_{ih} + \sum \limits_{h \in V_i(r_j): \delta_{ij} > \delta_{ih}} w_{ih} -  \epsilon^2 W_i(r_j)  \right) \\
\end{align*}

\item \textbf{Job $j$ is not dispatched to machine $i$}. Thus, there is no change in the queue on the machine $i$ at $r_j$ that is, $U_i(r_j) = U_i(r_j^-)$. Moreover, we have that $W_i(r_j) = W'_i(r_j) + \sum \limits_{h \in R^{2'}_i(r_j)} w_h/\epsilon - w_j$. Dividing the second term by $\epsilon$ and rearranging, we get 
$$
\sum \limits_{h \in R^{2'}_i(r_j)} w_h + \epsilon^2 W'_i(r_j) \geq  \epsilon^2 W_i(r_j) + \epsilon^2 w_j
$$ 
In this case, the Inequality~\ref{feasible-1} can be rewritten as: 
\begin{align*}
&\frac{\alpha_{ij}}{p_{ij}} \\  
&\leq  \frac{\epsilon}{1+\epsilon} \left(\frac{20 w_j}{\epsilon} + \frac{w_j}{p_{ij}} \sum \limits_{h \in V_i(r_j): \delta_{ih} \geq \delta_{ij}} p_{ih} + w_j + \sum \limits_{h \in V_i(r_j): \delta_{ij} > \delta_{ih}} w_{ih} - \sum \limits_{h \in R^{2'}_i(r_j)} w_{ih} - \epsilon^2 W'_i(r_j) \right) \\
 &\leq  \frac{\epsilon}{1+\epsilon} \left(\frac{20  (1+\epsilon)w_j}{\epsilon} + \frac{w_j}{p_{ij}} \sum \limits_{h \in V_i(r_j): \delta_{ih} \geq \delta_{ij}} p_{ih} + \sum \limits_{h \in V_i(r_j): \delta_{ij} > \delta_{ih}} w_{ih}  -  \epsilon^2 W_i(r_j)  \right) 
\end{align*}
\end{enumerate}
Thus the claim follows. 
\end{proof}
\medskip


Next, we prove the dual feasibility. We have two sub-cases depending on the density of the job running at some time $t' > t$
\smallskip

\noindent\textbf{Case 1}:
Let $z$ is executed at time $t'$ such that $\delta_{iz} \geq \delta_{ij}$ and $z \not= \kappa_i(r_j^-)$. Then we have 
\begin{align}
t'-r_j &\geq \sum \limits_{h \in V_i(r_j): \delta_{ih} > \delta_{iz}} p_{ih} + p_{iz} - q_{iz}(t') \label{t'-r_j-1}
\end{align}

From Inequality (\ref{feasible-A}), we have
\begin{align*}
\frac{\alpha_j}{p_{ij}} 
&\leq \frac{\epsilon}{1+\epsilon} \biggl(\frac{20  (1+\epsilon)w_j}{\epsilon} + \frac{w_j}{p_{ij}} \biggl( \sum \limits_{h \in V_i(r_j): \delta_{ih} > \delta_{iz}} p_{ih} + p_{iz} + \sum \limits_{h \in V_i(r_j): \delta_{ij} \leq \delta_{ih} < \delta_{iz}} p_{ih} \biggr) \\
 & \qquad \qquad \qquad + \sum \limits_{h \in V_i(r_j): \delta_{ij} > \delta_{ih}} w_{ih} - \epsilon^2 W_i(r_j) \biggr) \\
&\leq \frac{\epsilon}{1+\epsilon} \biggl(\frac{20  (1+\epsilon)w_j}{\epsilon} + \frac{w_j}{p_{ij}} \biggl( t' - r_j + q_{iz}(t) + \sum \limits_{h \in V_i(r_j): \delta_{ij} \leq \delta_{ih} < \delta_{iz}} p_{ih} \biggr) \\
 & \qquad \qquad  \qquad + \sum \limits_{h \in V_i(r_j): \delta_{ij} > \delta_{ih}} w_{ih} - \epsilon^2 W_i(r_j) \biggr) 
 \qquad \text{this follows from Inequality (\ref{t'-r_j-1})}\\
&\leq \frac{\epsilon}{1+\epsilon} \left(\frac{20  (1+\epsilon)w_j}{\epsilon} + \frac{w_j}{p_{ij}} ( t' - r_j ) + \left(\frac{w_{z}}{p_{i,z}} q_{i,z}(t') + \sum \limits_{h \in V_i(t')}w^f_{h}(t') \right) - \epsilon^2 W_i(r_j) \right) \\
&\leq \frac{\epsilon}{1+\epsilon} \left(\frac{20  (1+\epsilon)w_j}{\epsilon} + \frac{w_j}{p_{ij}} ( t' - r_j ) + \frac{1+\epsilon^2}{1+\epsilon^2} \left(\frac{w_{z}}{p_{i,z}} q_{i,z}(t') + \sum \limits_{h \in V_i(t')}w^f_{h}(t') \right) - \epsilon^2 W_i(r_j) \right) \\
&\leq \frac{\epsilon}{1+\epsilon} \left(\frac{20  (1+\epsilon)w_j}{\epsilon} + \frac{w_j}{p_{ij}} ( t' - r_j ) + \frac{1}{1+\epsilon^2} \left( \frac{w_{z}}{p_{iz}} q_{iz}(t') + \sum \limits_{h \in V_i(t')}w^f_{h}(t') \right)  \right)\\
&\qquad \qquad + \frac{\epsilon}{1+\epsilon} \left(\frac{\epsilon^2}{1+\epsilon^2} \left(\frac{w_{z}}{p_{iz}} q_{iz}(t')  + \sum \limits_{h \in V_i(t')}w^f_{h}(t') \right) - \epsilon^2 W_i(r_j) \right) \\
&\leq \frac{\epsilon}{1+\epsilon} \left(\frac{20  (1+\epsilon)w_j}{\epsilon} + \frac{w_j}{p_{ij}} ( t' - r_j ) + \frac{1}{1+\epsilon^2} \left( \frac{w_{z}}{p_{iz}} q_{iz}(t')  + \sum \limits_{h \in V_i(t')}w^f_{h}(t') \right)  \right)\\
&\qquad \qquad + \frac{\epsilon}{1+\epsilon} \left(\frac{\epsilon^2}{1+\epsilon^2} \left( W_i(t') + \frac{1}{\epsilon} \sum \limits_{h \in R_i(t)} w^f_h(t) \right) - \epsilon^2 W_i(r_j) \right) \qquad \text{due to Lemma~\ref{lem:main-inequality}}\\
&\leq \frac{\epsilon}{1+\epsilon} \left(\frac{20  (1+\epsilon)w_j}{\epsilon} + \frac{w_j}{p_{ij}} ( t' - r_j ) + \frac{1}{1+\epsilon^2} \left( \frac{w_{z}}{p_{iz}} q_{iz}(t') + \sum \limits_{h \in V_i(t')}w^f_{h}(t') + \sum \limits_{h \in R_i(t')}w^f_{h}(t') \right)  \right)\\
&\qquad \qquad + \frac{\epsilon}{1+\epsilon} \left(\frac{\epsilon^2}{1+\epsilon^2} W_i(t') - \epsilon^2 W_i(r_j) \right) \\
&\leq \frac{\epsilon}{1+\epsilon} \left(\frac{20  (1+\epsilon)w_j}{\epsilon} + \frac{w_j}{p_{ij}} ( t' - r_j ) + \frac{1}{1+\epsilon^2} \left(\sum \limits_{h \in Q_i(t')}w^f_{h}(t') \right)  \right)\\
&\leq 20 w_j+ \frac{w_j}{p_{ij}} ( t' - r_j ) + \beta_{it'} 
\end{align*}

\noindent\textbf{Case 2}: Let $z$ is executed at time $t'$ such that $\delta_{iz} \geq \delta_{ij}$ or $z = \kappa_i(r_j^-)$. A set of similar arguments show that the feasibility holds.
\end{proof}

\begin{lemma} \label{feasible-dual-constraints-j-nu-rejected}
Assume that a job $j$ is immediately rejected at $r_j$ and $R^2_i(r_j) = \{j, \nu_i(r_j^-)\}$  when $j$ is hypothetically assigned to $i$. Then, the dual constraint (\ref{dual-const}) corresponding to $j$ holds.
\end{lemma}

\begin{proof}
Fix a machine $i$ and the time $t$. As before, we prove the above inequality at $r_j$, assuming that no job arrives after $r_j$. Let $R_i(r_j)$ denote the set of job rejected due to arrival of $j$, that is $R_i(r_j) = R^1_i(r_j) \cup R^2_i(r_j)$. Using the definition of $\alpha_{ij}$, we have 
\begin{align*}
\frac{\alpha_j}{p_{ij}} &\leq  \left(\frac{\epsilon}{1+\epsilon} \right) \frac{\alpha_{ij}}{p_{ij}} \\
  &=  \frac{\epsilon}{1+\epsilon} \left(\frac{20 w_j}{\epsilon} + \frac{w_j}{p_{ij}} \sum \limits_{h \in V_i(r_j^-): \delta_{ih} \geq \delta_{ij}} p_{ih} + w_j + \sum \limits_{h \in V_i(r_j^-): \delta_{ij} > \delta_{ih}} w_{ih} - \frac{w_j}{p_{ij}} \sum \limits_{h \in R^{2'}_i(r_j)} p_{ih}\right) 
\end{align*}
Note that all jobs in $V_i(r_j^-)$ have density smaller than the job $j$. Therefore, the above inequality can be re-written as:
\begin{align}
\frac{\alpha_j}{p_{ij}} \leq   \frac{\epsilon}{1+\epsilon} \left(\frac{20 w_j}{\epsilon} + \frac{w_j}{p_{ij}} \sum \limits_{h \in V_i(r_j^-): \delta_{ih} \geq \delta_{ij}} p_{ih} + w_j  -  \frac{w_j}{p_{ij}} \sum \limits_{h \in R^{2'}_i(r_j)} p_{ih}\right) \label{feasible-B}
\end{align}

\noindent\textbf{Case A}: Assume that the algorithm assigns $j$ to $i$. 
The set of jobs in $R^{2'}_i(r_j)$ are removed and $U_i(r_j) = (U_i(r_j^-) \cup \{j\})\setminus R_i(r_j)$.  From Property 1 in Lemma~\ref{wg-rule-prop} that $W_i(r_j) = 0$. We can rewrite Inequality~\ref{feasible-B} as:

\begin{align*}
\frac{\alpha_j}{p_{ij}} &\leq \frac{\epsilon}{1+\epsilon} \left(\frac{20  (1+\epsilon)w_j}{\epsilon} + \frac{w_j}{p_{ij}} \sum \limits_{h \in V_i(r_j): \delta_{ih} \geq \delta_{ij}} p_{ih}  - \epsilon^2 W_i(r_j) \right)
\end{align*}
where $W_i(r_j) = 0$. The dual feasibility for this sub-case can be shown similar to Lemma~\ref{feasible-dual-constraints-j-not-rejected}.
\medskip

\noindent\textbf{Case B}: Now, suppose that $j$ is not assigned to $i$, then $U_i(r_j^-) = U_i(r_j)$. 
Hence, we can rewrite Inequality~\ref{feasible-B} as:

\begin{align*}
\frac{\alpha_j}{p_{ij}} &\leq \frac{\epsilon}{1+\epsilon} \left(\frac{20 (1+\epsilon)w_j}{\epsilon} + \frac{w_j}{p_{ij}} \sum \limits_{h \in V_i(r_j)\setminus \nu_i(r_j^-) : \delta_{ih} \geq \delta_{ij}} p_{ih}\right) 
\end{align*}
\medskip

Below, we present the dual feasibility proof of this sub-case. Depending on the density of the job running at some time $t' > t$, we have two sub-cases:

\noindent\textbf{Case 1}: Let $t' > t$ and $z$ is job executed at $t'$. Note that $\delta_{iz} > \delta_{i,\nu_i(r_j^-)}$. Thus we have that:
\begin{align*}
&\frac{\alpha_j}{p_{ij}} \leq \frac{\epsilon}{1+\epsilon} \left(\frac{20  (1+\epsilon)w_j}{\epsilon} + \frac{w_j}{p_{ij}} \sum \limits_{h \in V_i(r_j): \delta_{ih} > \delta_{i,\nu_i(r_j^-)}} p_{ih}\right) \\
 &\leq \frac{\epsilon}{1+\epsilon} \left(\frac{20  (1+\epsilon)w_j}{\epsilon} + \frac{w_j}{p_{ij}} \left(\sum \limits_{h \in V_i(r_j): \delta_{ih} > \delta_{iz}} p_{ih} + p_{iz} + \sum \limits_{h \in V_i(r_j): \delta_{iz} > \delta_{ih} > \delta_{i,\nu_i(r_j^-)}} p_{ih}\right) \right)\\
&\leq \frac{\epsilon}{1+\epsilon} \left(\frac{20  (1+\epsilon)w_j}{\epsilon} + \frac{w_j}{p_{ij}} \left(t' - r_j + q_{iz}(t') + \sum \limits_{h \in U_i(r_j)\setminus k_i(r_j^-) : \delta_{iz} > \delta_{ih} > \delta_{i\ell_i(r_j)}} p_{ih}\right) \right)\\
&\leq \frac{\epsilon}{1+\epsilon} \left(\frac{20  (1+\epsilon)w_j}{\epsilon} + \frac{w_j}{p_{ij}} (t' - r_j) + \frac{w_j}{p_{ij}} q_{iz}(t') + \sum \limits_{h \in V_i(r_j): \delta_{iz} > \delta_{ih} > \delta_{i,\nu_i(r_j^-)}} w_{ih} + w_{\nu_i(r_j^-)} - w_{\nu_i(r_j^-)} \right)\\
&\leq \frac{\epsilon}{1+\epsilon} \left(\frac{20  (1+\epsilon)w_j}{\epsilon} + \frac{w_j}{p_{ij}} (t' - r_j) +  \frac{w_j}{p_{ij}} q_{iz}(t')  + \sum \limits_{h \in V_i(t')} w^f_{h}(t)  - w_{\nu_i(r_j^-)} \right)\\
&\leq \frac{\epsilon}{1+\epsilon} \left(\frac{20  (1+\epsilon)w_j}{\epsilon} + \frac{w_j}{p_{ij}} (t' - r_j) +  \frac{w_j}{p_{ij}} q_{iz}(t')  + \sum \limits_{h \in V_i(t')} w^f_{h}(t)  - \epsilon^2 W_i(r_j) \right)\\
&\leq \frac{\epsilon}{1+\epsilon} \left(\frac{20  (1+\epsilon)w_j}{\epsilon} + \frac{w_j}{p_{ij}} (t' - r_j) +  \frac{1+\epsilon^2}{1+\epsilon^2} \left( \frac{w_z}{p_{iz}} q_{iz}(t')  + \sum \limits_{h \in V_i(t')} w^f_{h}(t) \right) - \epsilon^2 W_i(r_j) \right)\\
&\leq \frac{\epsilon}{1+\epsilon} \left(\frac{20  (1+\epsilon)w_j}{\epsilon} + \frac{1}{1+\epsilon^2} \left(\sum \limits_{h \in Q_i(t')} w^f_{h}(t) \right) + \frac{\epsilon^2}{1+\epsilon^2} W_i(r_j) - \epsilon^2 W_i(r_j) \right)\\
&\leq 20 w_j + \frac{w_j}{p_{ij}} (t'-r_j) + \beta_{it'}
\end{align*}

\noindent \textbf{Case 2}: Let $t' > t$ such that job $\nu_i(r_j^-)$ is executing.  We have that $t'-r_j \geq \sum \limits_{h \in V_i(r_j)\cup \nu_i(r_j^-)) : \delta_{ih} \geq \delta_{ij}} p_{ih}$. Thus, we have:
\begin{align*}
\frac{\alpha_j}{p_{ij}} &\leq \frac{\epsilon}{1+\epsilon} \left(\frac{20  (1+\epsilon)w_j}{\epsilon} + \frac{w_j}{p_{ij}} \sum \limits_{h \in V_i(r_j) : \delta_{ih} > \delta_{i,\nu_i(r_j^-)}} p_{ih}\right) \\
&\leq \frac{\epsilon}{1+\epsilon} \left(\frac{20  (1+\epsilon)w_j}{\epsilon} + \frac{w_j}{p_{ij}} (t'-r_j) \right)
\leq 20 w_j + \frac{w_j}{p_{ij}} (t'-r_j) + \beta_{it'}
\end{align*}
\end{proof}

\begin{lemma} \label{feasible-dual-constraints-j-only-rejected}
	Suppose that a job $j$ is immediately rejected at $r_j$ and $R^2_i(r_j) = \{j\}$  when $j$ is hypothetically assigned to $i$.  Then, the dual constraint (\ref{dual-const}) corresponding to $j$ holds.
\end{lemma}

\begin{proof}
Fix a machine $i$ and the time $t$. Again, we prove the above inequality at $r_j$, assuming that no job arrives after $r_j$. Let $R_i(r_j)$ denote the set of job rejected due to arrival of $j$, that is $R_i(r_j) = R^1_i(r_j) \cup R^2_i(r_j)$. Define $\rho = \rho_{ij}$ at $V_i(r_j^-)$. Using the definition of $\alpha_{ij}$, we have
\begin{align}
\frac{\alpha_{j}}{p_{ij}} &\leq \frac{\alpha_{ij}}{p_{ij}} = \frac{\epsilon} {1+\epsilon} \left(\frac{20 w_j}{\epsilon} + \frac{w_j}{p_{ij}} \sum \limits_{h \in V_i(r_j^-)} p_{ih} + w_j\right)  \nonumber\\
& \qquad \qquad - \frac{\epsilon} {1+\epsilon} \left(\frac{w_j}{p_{ij}} \left(\sum \limits_{h \in V_i(r_j^-) :  \delta_{ih} \leq \delta_{i\rho}}  p_h + \left(\frac{W'_i(r_j) -  \sum \limits_{h \in V_i(r_j^-) :  \delta_{ih} \leq \delta_{\rho}} w_h}{w_{\rho-1}}\right)p_{i,(\rho-1)}\right) \right) \nonumber \\
&\leq \frac{\epsilon} {1+\epsilon} \left(\frac{20 w_j}{\epsilon} + \frac{w_j}{p_{ij}} \sum \limits_{h \in V_i(r_j^-): \delta_{ih} \geq \delta_{i,(\rho-1)}} p_{ih} + w_j -\frac{w_j}{p_{ij}} \left(\frac{W'_i(r_j) -  \sum \limits_{h \in V_i(r_j^-):  \delta_{ih} \leq \delta_{i\rho}} w_h}{w_{\rho-1}}\right)p_{i,(\rho-1)}\right) \label{j-only-rejected}
\end{align}

Below, we present the dual feasibility proof. Depending on the density of the job running at some time $t' > t$, we consider cases. \\
 
\noindent\textbf{Case 1}: Let $t' > t$ and $z$ is executed at time $t'$. Recall that all jobs in $V_i(r_j^-)$ have density smaller than $j$. 
Assume that 
$$
t' - r_j \leq \sum \limits_{h \in V_i(r_j^-): \delta_{ih} \geq \delta_{i,(\rho-2)}} p_{ih} + \left(1- \frac{W'_i(r_j) -  \sum \limits_{h \in V_i(r_j^-):  \delta_{ih} \leq \delta_{i\rho}} w_h}{w_{\rho-1}}\right)p_{i,(\rho-1)}
$$ 
Thus, $\delta_{iz} \geq \delta_{i,(\rho-1)}$.

\noindent \textbf{Case 1.1}: Assume that $z \not = \rho-1$, then we have: 
\begin{align*} 
& \alpha_j/p_{ij} \\
&\leq \frac{\epsilon} {1+\epsilon} \left(\frac{20 w_j}{\epsilon} + \frac{w_j}{p_{ij}} \left(\sum \limits_{h \in V_i(r_j^-): \delta_{ih} \geq \delta_{i,(\rho-1)}} p_{ih} - \left(\frac{W'_i(r_j) -  \sum \limits_{h \in V_i(r_j^-):  \delta_{ih} \leq \delta_{i\rho}} w_h}{w_{\rho-1}}\right)p_{i,(\rho-1)}\right)+ w_j \right)\\
&\leq \frac{\epsilon} {1+\epsilon} \left(\frac{20 w_j}{\epsilon} + \frac{w_j}{p_{ij}} \left(\sum \limits_{h \in V_i(r_j): \delta_{ih} > \delta_{iz}} p_{ih} + p_{iz} +\sum \limits_{h \in V_i(r_j): \delta_{i,(\rho-1)} \leq \delta_{ih} < \delta_{iz}} p_{ih} \right)+ w_j \right) \\
&\qquad \qquad -  \frac{\epsilon} {1+\epsilon} \left(\frac{w_j}{p_{ij}}\right) \left(\frac{W'_i(r_j) -  \sum \limits_{h \in V_i(r_j^-):  \delta_{ih} \leq \delta_{i\rho}} w_h}{w_{\rho-1}}\right)p_{i,(\rho-1)}\\
&\leq \frac{\epsilon} {1+\epsilon} \left(\frac{20 w_j}{\epsilon} + \frac{w_j}{p_{ij}} \left(t' - r_j + q_{iz}(t') + \sum \limits_{h \in V_i(r_j): \delta_{i,(\rho-2)} \leq \delta_{ih} < \delta_{iz}} p_{ih} \right)+ w_j \right) \\
&\qquad \qquad+  \frac{\epsilon} {1+\epsilon} \left(\frac{w_j}{p_{ij}} \left( 1 - \frac{W'_i(r_j) -  \sum \limits_{h \in V_i(r_j^-):  \delta_{ih} \leq \delta_{i\rho}} w_h}{w_{\rho-1}}\right)p_{i,(\rho-1)}  \right)\\
&\leq \frac{\epsilon} {1+\epsilon} \left(\frac{20 w_j}{\epsilon} + \frac{w_j}{p_{ij}} \left(t' - r_j + q_{iz}(t') + \sum \limits_{h \in V_i(r_j): \delta_{i,(\rho-2)} \leq \delta_{ih} < \delta_{iz}} p_{ih} \right)+ w_j \right) \\
&\qquad \qquad+  \frac{\epsilon} {1+\epsilon} \left(\sum \limits_{h \in V_i(r_j^-):  \delta_{ih} \leq \delta_{i,(\rho-1)}} w_h - W'_i(r_j) \right) \\
&\leq \frac{\epsilon} {1+\epsilon} \left(\frac{20 w_j}{\epsilon} + \frac{w_j}{p_{ij}}(t' - r_j) + \frac{w_j}{p_{ij}} \left(q_{iz}(t') + \sum \limits_{h \in V_i(r_j): \delta_{ih} < \delta_{iz}} p_{ih} \right)+ w_j - W'_i(r_j) \right)
\end{align*}

Recall that $W_i(r_j^-) - W_i'(r_j) \leq w_j/\epsilon$. Using this above inequality can be re-written as:
\begin{align*}
&\alpha_j/p_{ij} \\
&\leq \frac{\epsilon} {1+\epsilon} \left(\frac{20 w_j}{\epsilon} + \frac{w_j}{p_{ij}}(t' - r_j) + \frac{w_j}{p_{ij}} \left(q_{iz}(t') + \sum \limits_{h \in V_i(r_j): \delta_{ih} < \delta_{iz}} p_{ih} \right)+ w_j - W_i(r_j) + w_j /\epsilon \right) \\
&\leq \frac{\epsilon} {1+\epsilon} \left(\frac{21 w_j}{\epsilon} + \frac{w_j}{p_{ij}}(t' - r_j) + \frac{w_j}{p_{ij}} \left(q_{iz}(t') + \sum \limits_{h \in V_i(r_j): \delta_{ih} < \delta_{iz}} p_{ih} \right)+ w_j - W_i(r_j) \right) \\
&\leq \frac{\epsilon} {1+\epsilon} \left(\frac{21 (1+\epsilon) w_j}{\epsilon} + \frac{w_j}{p_{ij}}(t' - r_j) +  \sum \limits_{h \in Q_i(t')} w^f_{h} - W_i(r_j) \right) \\
&\leq 21 w_j+ \frac{w_j}{p_{ij}}(t'-r_j) + \beta_{it}
\end{align*}

\noindent \textbf{Case 1.2}: Assume that $z  = \rho-1$, then we have from Inequality~\ref{j-only-rejected}:

\begin{align*}
&\alpha_j/p_{ij} \\
&\leq \frac{\epsilon} {1+\epsilon} \left(\frac{20 w_j}{\epsilon} + \frac{w_j}{p_{ij}} \left(\sum \limits_{h \in V_i(r_j^-): \delta_{ih} \geq \delta_{i,(\rho-1)}} p_{ih} - \left(\frac{W'_i(r_j) -  \sum \limits_{h \in V_i(r_j^-):  \delta_{ih} \leq \delta_{i\rho}} w_h}{w_{\rho-1}}\right)p_{i,(\rho-1)}\right)+ w_j \right)\\
&\leq \frac{\epsilon} {1+\epsilon} \left(\frac{20 w_j}{\epsilon} + \frac{w_j}{p_{ij}} \left(t' - r_j + q_{iz}(t') - p_{i,z} \right)+ w_j \right) \\
&\qquad \qquad +  \frac{\epsilon} {1+\epsilon} \left(\left( \frac{ w_{\rho-1}.q_{i,(\rho-1)}(t')}{p_{i,(\rho-1)}} - (W'_i(r_j) -  \sum \limits_{h \in V_i(r_j^-):  \delta_{ih} \leq \delta_{i\rho}} w_h)\right)- \frac{w_j}{p_{ij}} q_{i,(\rho-1)} (t') \right)\\
&\leq \frac{\epsilon} {1+\epsilon} \left(\frac{20 w_j}{\epsilon} + \frac{w_j}{p_{ij}} \left(t' - r_j \right)+ w_j + \frac{ w_{\rho-1}.q_{i,(\rho-1)}(t')}{p_{i,(\rho-1)}} - (W'_i(r_j) -  \sum \limits_{h \in V_i(r_j^-):  \delta_{ih} \leq \delta_{i\rho}} w_h)\right)\\
&\leq \frac{\epsilon} {1+\epsilon} \left(\frac{20 w_j}{\epsilon} + \frac{w_j}{p_{ij}} \left(t' - r_j \right)+ w_j + \frac{ w_{z}.q_{i,z}(t')}{p_{i,z}} + \sum \limits_{h \in V_i(r_j^-):  \delta_{ih} \leq \delta_{i\rho}} w_h - W'_i(r_j) \right)\\
\end{align*}

Recall that $W_i(r_j^-) - W_i(r_j^+) \leq w_j/\epsilon$. Using this above inequality can be re-written as:

\begin{align*}
&\alpha_j/p_{ij} \\ 
&\leq \frac{\epsilon} {1+\epsilon} \left(\frac{20 w_j}{\epsilon} + \frac{w_j}{p_{ij}} \left(t' - r_j \right)+ w_j + \frac{ w_{z}.q_{i,z}(t')}{p_{i,z}} + \sum \limits_{h \in V_i(r_j^-):  \delta_{ih} \leq \delta_{i\rho}} w_h - W_i(r_j) + \frac{w_j}{\epsilon} \right)\\
&\leq \frac{\epsilon} {1+\epsilon} \left(\frac{21 (1+\epsilon) w_j}{\epsilon} + \frac{w_j}{p_{ij}} (t'-r_j) + \frac{1+\epsilon^2}{1+\epsilon^2} \left(\frac{ w_{z}.q_{i,z}(t')}{p_{i,z}}  + \sum\limits_{h \in V_i(t')} w^f_{h}(t')\right) - W_i(r_j) \right) \\
&\leq \frac{\epsilon} {1+\epsilon} \left(\frac{21 (1+\epsilon) w_j}{\epsilon} + \frac{w_j}{p_{ij}} (t'-r_j) + \frac{\sum\limits_{h \in Q_i(t')} w_f(t')}{1+\epsilon^2}  + \frac{\epsilon^2}{1+\epsilon^2} W_i(r_j) - W_i(r_j) \right) \\
&\leq 21w_j + \frac{w_j}{p_{ij}} (t'-r_j) + \beta_{it'} \\
\end{align*}

\noindent \textbf{Case 2} Let $t' > t$ and $t'-r_j > \sum \limits_{h \in V_i(r_j^-): \delta_{ih} \geq \delta_{i,(\rho-2)}} p_{ih} + \left(1- \frac{W'_i(r_j) -  \sum \limits_{h \in V_i(r_j^-):  \delta_{ih} \leq \delta_{i\rho}} w_h}{w_{\rho-1}}\right)p_{i,(\rho-1)}$. Then we can re-write the Inequality~\ref{j-only-rejected} as:

\begin{align*}
&\alpha_j/p_{ij} \\ 
&\leq \frac{\epsilon} {1+\epsilon} \left(\frac{20 w_j}{\epsilon} + \frac{w_j}{p_{ij}} \sum \limits_{h \in V_i(r_j^-): \delta_{ih} \geq \delta_{i,(\rho-1)}} p_{ih} + w_j -\frac{w_j}{p_{ij}} \left(\frac{W'_i(r_j) -  \sum \limits_{h \in V_i(r_j^-):  \delta_{ih} \leq \delta_{i\rho}} w_h}{w_{\rho-1}}\right)p_{i,(\rho-1)}\right) \\
&\leq \frac{\epsilon} {1+\epsilon} \left(\frac{20 (1+\epsilon) w_j}{\epsilon} + \frac{w_j}{p_{ij}} \left(\sum \limits_{h \in V_i(r_j^-): \delta_{ih} \geq \delta_{i,(\rho-2)}} p_{ih} + \left(1 - \frac{W'_i(r_j) -  \sum \limits_{h \in V_i(r_j^-):  \delta_{ih} \leq \delta_{i\rho}} w_h}{w_{\rho-1}}\right)p_{i,(\rho-1)}\right) \right) \\
&\leq \frac{\epsilon} {1+\epsilon} \left(\frac{20 (1+\epsilon) w_j}{\epsilon} + \frac{w_j}{p_{ij}} (t'-r_j) \right) 
\end{align*}
and the lemma holds
\end{proof}

\begin{lemma} \label{bound-alpha}
It holds that $\sum\limits_{j \in \mathcal{J}} \alpha_j \geq \frac{\epsilon}{1+\epsilon}\sum \limits_{j \in \mathcal{J}} (\widetilde{C}_j - r_j)$.
\end{lemma}

\begin{proof}

Using the definitions of $\alpha_j$, we have that 

\begin{align*}
&\frac{1+\epsilon}{\epsilon} \sum \limits_{j} \alpha_j  =  \frac{1+\epsilon}{\epsilon}  \left(\sum \limits_{j \not \in R^2_i(r_j)} \alpha_j + \sum \limits_{j \in R^2_{ij}} \left(\alpha_j.\one_{|R^2_i(r_j) = 1|} +  \alpha_j. \one_{|R^2_{ij} > 1|} \right) \right)\\
&=  \sum_j \left( \frac{20 w_j p_{ij}}{\epsilon} + w_j \sum \limits_{h \in V_i(r_j^-): \delta_{ih} \geq \delta_{ij}} p_{ih} + w_j p_{ij} + p_{ij}\sum \limits_{h \in V_i(r_j^-): \delta_{ij} > \delta_{ih}} w_{ih} \right) \\
&\qquad \qquad- \sum \limits_{R^2_i(r_j) = \{j\}} w_j \left(\sum \limits_{h \in V_i(r_j^-):  \delta_{ih} \leq \delta_{a_{ij}}}  p_h + \left(\frac{W_i(r_j) -  \sum \limits_{h \in V_i(r_j^-):  \delta_{ih} \leq \delta_{i,\rho_{ij}}} w_h}{w_{\rho_{ij}-1}}\right)p_{i,(\rho_{ij}-1)}\right)\\
&\qquad \qquad- \sum \limits_{R^2_{ij} = \{j, \nu_i(r_j^-)\}} w_j \left(\sum \limits_{h \in R^2_{ij}} p_{ih}\right) - \sum \limits_{j \not \in R^2_i(r_j)} \left(p_{ij} \sum \limits_{h \in R^2_i(r_j)} w_h +  \epsilon^2 W_i(r_j) p_{ij} \right)\\ 
&\geq   \sum_j \left( \frac{19 w_j p_{ij}}{\epsilon} + w_j p_{i,\kappa_i(r_j)} + w_j \sum \limits_{h \in V_i(r_j^-): \delta_{ih} \geq \delta_{ij}} p_{ih} + w_j p_{ij} + p_{ij}\sum \limits_{h \in V_i(r_j^-) : \delta_{ij} > \delta_{ih}} w_{ih} \right) \\
&\qquad \qquad- \sum \limits_{R^2_i(r_j) = \{j\}} w_j \left(\sum \limits_{h \in V_i(r_j^-):  \delta_{ih} \leq \delta_{\rho_{ij}}}  p_h + \left(\frac{W_i(r_j) -  \sum \limits_{h \in V_i(r_j^-):  \delta_{ih} \leq \delta_{\rho_{ij}}} w_h}{w_{\rho_{ij}-1}}\right)p_{\rho_{ij}-1}\right)\\
&\qquad \qquad- \sum \limits_{R^2_i(r_j) = \{j,\nu_i(r_j^-)\}} w_j \left(\sum \limits_{h \in R^2_i(r_j)} p_{ih}\right) - \sum \limits_{j \not \in R^2_i(r_j)} \left(p_{ij} \sum \limits_{h \in R^2_i(r_j)} w_h +  \epsilon^2 W_i(r_j) p_{ij} \right)\\ 
\end{align*}

The above inequality holds due to \emph{preempt} rule that the total weight of jobs arriving due a time $j$ is executing on a machine $i$ is at most $w_j/\epsilon$. 

Now we split the above inequality into three and show that each part can be bounded separately, that is 
\begin{align*}
\frac{1+\epsilon}{\epsilon} \sum_j \alpha_j &\geq A^1 + A^2 + A^3
\end{align*}
where 
\begin{align}
A^1 &:=  \sum_{R^2_i(r_j) = \{j\}} \left( \frac{19 w_j p_{ij}}{\epsilon} + w_j p_{i,\kappa_i(r_j)} + w_j \sum \limits_{h \in V_i(r_j^-): \delta_{ih} \geq \delta_{ij}} p_{ih} + w_j p_{ij} + p_{ij}\sum \limits_{h \in V_i(r_j^-): \delta_{ij} > \delta_{ih}} w_{ih} \right) \nonumber \\
&\qquad \qquad - \sum \limits_{R^2_i(r_j) = \{j\}} w_j \left(\sum \limits_{h \in V_i(r_j^-):  \delta_{ih} \leq \delta_{\rho_{ij}}}  p_h + 
	\left(\frac{W_i(r_j) -  \sum \limits_{h \in V_i(r_j^-):  \delta_{ih} \leq \delta_{\rho_{ij}}} w_h}{w_{\rho_{ij}-1}}\right)p_{\rho_{ij}-1}\right), \label{def-A1}\\
\nonumber \\
A^2 &:=  \sum \limits_{R^2_i(r_j) = \{j,\nu_i(r_j^-)\}} \left( \frac{19 w_j p_{ij}}{\epsilon} + w_j p_{i,\kappa_i(r_j)} + w_j \sum \limits_{h \in V_i(r_j^-): \delta_{ih} \geq \delta_{ij}} p_{ih} + w_j p_{ij} + p_{ij}\sum \limits_{h \in V_i(r_j^-): \delta_{ij} > \delta_{ih}} w_{ih} \right) \nonumber \\
&\qquad \qquad - \sum \limits_{R^2_i(r_j) = \{j,\nu_i(r_j^-)\}} w_j \left(\sum \limits_{h \in R^2_i(r_j)} p_{ih}\right), \label{def-A2} \\ \text{~and~} \nonumber
 \\
 A^3 &:= \sum \limits_{j \not \in R^2_i(r_j)}  \left( \frac{19 w_j p_{ij}}{\epsilon} + w_j p_{i,\kappa_i(r_j)} + w_j \sum \limits_{h \in V_i(r_j^-): \delta_{ih} \geq \delta_{ij}} p_{ih} + w_j p_{ij} + p_{ij}\sum \limits_{h \in V_i(r_j^-): \delta_{ij} > \delta_{ih}} w_{ih} \right) \nonumber \\
&\qquad \qquad -  \sum \limits_{j \not \in R^2_i(r_j)} \left(p_{ij} \sum \limits_{h \in R^2_i(r_j)} w_h +  \epsilon^2 W_i(r_j) p_{ij} \right). \label{def-A3}
\end{align}
We now show individual lower bounds on $A^1, A^2$ and $A^3$. 
\bigskip

\begin{claim} \label{lb-A1}
$A^1 \geq \sum\limits_{R^2_i(r_j) = \{j\}} w_j (\widetilde{C}_j - r_j) + 19w_j p_{ij}/\epsilon$.
\end{claim}

\begin{proof}
Using the definition of $A^1$ from the Equation~\ref{def-A1}, we have 
\begin{align*}
A^1 &= \sum_{R^2_i(r_j) = \{j\}} \left( \frac{19 w_j p_{ij}}{\epsilon} + w_j p_{i,\kappa_i(r_j)} + w_j \sum \limits_{h \in V_i(r_j^-): \delta_{ih} \geq \delta_{ij}} p_{ih} + w_j p_{ij} + p_{ij}\sum \limits_{h \in V_i(r_j^-): \delta_{ij} > \delta_{ih}} w_{ih} \right) \\
& \qquad - \sum \limits_{R^2_i(r_j) = \{j\}} w_j \left(\sum \limits_{h \in V_i(r_j^-):  \delta_{ih} \leq \delta_{i,\rho_{ij}}}  p_h + 
	\left(\frac{W_i(r_j) -  \sum \limits_{h \in V_i(r_j^-):  \delta_{ih} \leq \delta_{i,\rho_{ij}}} w_h}{w_{\rho_{ij}-1}}\right)p_{i,\rho_{ij}-1}\right)\\
&= \sum_{R^2_i(r_j) = \{j\}} \left( \frac{19 w_j p_{ij}}{\epsilon} + w_j p_{i,\kappa_i(r_j)} + w_j \sum \limits_{h \in V_i(r_j): \delta_{ih} \geq \delta_{ij}} p_{ih} + w_j p_{ij} \right) \\
& \qquad - \sum \limits_{R^2_i(r_j) = \{j\}} w_j \left(\sum \limits_{h \in V_i(r_j):  \delta_{ih} \leq \delta_{i,\rho_{ij}}}  p_h + 
	\left(\frac{W_i(r_j) -  \sum \limits_{h \in V_i(r_j):  \delta_{ih} \leq \delta_{i,\rho_{ij}}} w_h}{w_{\rho_{ij}-1}}\right)p_{i,\rho_{ij}-1}\right)\\
&\qquad \qquad \qquad \text{since $j$ is the smallest density job}\\
&= \sum_{R^2_i(r_j) = \{j\}} \left( \frac{19 w_j p_{ij}}{\epsilon} + w_j p_{i,\kappa_i(r_j)} + w_j \sum \limits_{h \in V_i(r_j): \delta_{ih} \geq \delta_{i,\rho_{ij}-1}} p_{ih} + w_j p_{ij} \right) \\
&\qquad - \sum \limits_{R^2_i(r_j) = \{j\}} w_j 
	\left(\frac{W_i(r_j) -  \sum \limits_{h \in V_i(r_j):  \delta_{ih} \leq \delta_{i,\rho_{ij}}} w_h}{w_{\rho_{ij}-1}}\right)p_{i,\rho_{ij}-1}\\
&= \sum_{R^2_i(r_j) = \{j\}} \left( \frac{19 w_j p_{ij}}{\epsilon} + w_j p_{i,\kappa_i(r_j)} + w_j \sum \limits_{h \in V_i(r_j): \delta_{ih} \geq \delta_{i,\rho_{ij}-2}} p_{ih} + w_j p_{ij} \right) \\
&\qquad + \sum \limits_{R^2_i(r_j) = \{j\}} w_j 
	\left(\frac{\sum \limits_{h \in V_i(r_j):  \delta_{ih} \leq \delta_{i,\rho_{ij-1}}} w_h - W_i(r_j) }{w_{\rho_{ij}-1}}\right)p_{i,\rho_{ij}-1}\\
&= w_j \sum_{R^2_i(r_j) = \{j\}} \left( \frac{19 p_{ij}}{\epsilon} +  p_{i,\kappa_i(r_j)} + \sum \limits_{h \in V_i(r_j): \delta_{ih} \geq \delta_{i,\rho_{ij}-2}} p_{ih} + p_{ij} \right) \\
&\qquad + \sum \limits_{R^2_i(r_j) = \{j\}}
	\left(\frac{\sum \limits_{h \in V_i(r_j):  \delta_{ih} \leq \delta_{i,\rho_{ij-1}}} w_h - W_i(r_j) }{w_{\rho_{ij}-1}}\right)p_{i,\rho_{ij}-1}\\	
&= \sum \limits_{R^2_i(r_j) = \{j\}} w_j (\widetilde{C}_j - r_j) + 19 w_j p_{ij}/\epsilon \text{~~~~~~~~~~~~from the definition of $\widetilde{C}_j$}
\end{align*}
\end{proof}
\medskip 

\begin{claim} \label{lb-A2}
$A^2 \geq \sum \limits_{R^2_i(r_j) = \{j, \nu_i(r_j^-)\}} w_j (\widetilde{C}_j - r_j) + 19w_j p_{ij}/\epsilon$.
\end{claim}

\begin{proof}
Using the definition of $A^2$ from the Equation~\ref{def-A2}, we have:
\begin{align*}
A^2 &=  \sum \limits_{R^2_i(r_j) = \{j, \nu_i(r_j^-)\}} \left( \frac{19 w_j p_{ij}}{\epsilon} + w_j p_{i,\kappa_i(r_j)} + w_j \sum \limits_{h \in V_i(r_j^-) : \delta_{ih} \geq \delta_{ij}} p_{ih} + w_j p_{ij} + p_{ij}\sum \limits_{h \in V_i(r_j^-): \delta_{ij} > \delta_{ih}} w_{ih} \right) \\
 & \qquad - \sum \limits_{R^2_i(r_j) = \{j, \nu_i(r_j^-)\}} w_j \left(\sum \limits_{h \in R^2_i(r_j)} p_{ih}\right) \\
&=  \sum \limits_{R^2_i(r_j) = \{j, \nu_i(r_j^-)\}} \left( \frac{19 w_j p_{ij}}{\epsilon} + w_j p_{i,\kappa_i(r_j)} + w_j \sum \limits_{h \in V_i(r_j^-): \delta_{ih} \geq \delta_{ij}} p_{ih} + w_j p_{ij}  - w_j \sum \limits_{h \in R^2_i(r_j)} p_{ih}\right) \\
& \qquad \qquad \qquad \text{since $j$ is the smallest density job} \\
&=  \sum \limits_{R^2_i(r_j) = \{j, \nu_i(r_j^-)\}} \left( \frac{19 w_j p_{ij}}{\epsilon} + w_j p_{i,\kappa_i(r_j)} + w_j \sum \limits_{h \in V_i(r_j): \delta_{ih} \geq \delta_{ij}} p_{ih} + w_j p_{ij} \right) \\
&=  \sum \limits_{R^2_i(r_j) = \{j, \nu_i(r_j^-)\}} w_j \left( \frac{19 p_{ij}}{\epsilon} + w_j p_{i,\kappa_i(r_j)} +  \sum \limits_{h \in V_i(r_j): \delta_{ih} \geq \delta_{ij}} p_{ih} + p_{ij} \right) \\
&= \sum \limits_{R^2_i(r_j) = \{j, \nu_i(r_j^-)\}} w_j (\widetilde{C}_j - r_j) + 19w_j p_{ij}/\epsilon \text{~~~~~~~~~~~~from the definition of $\widetilde{C}_j$}
\end{align*}
\end{proof}
\medskip 

\begin{claim} \label{lb-A3}
$A^3 \geq \sum \limits_{j \not \in R^2_i(r_j)}  \left(w_j (\widetilde{C}_j - r_j) + \frac{15 w_j p_{ij}}{\epsilon}\right) - \frac{5}{\epsilon} \left( \sum \limits_{j \in \mathcal{J} } w_j p_{ij} \right)$
\end{claim}

\begin{proof}
Recall that the smallest density job in $R^2_i(r_j)$ is $\nu_i(r_j^-)$. From the Inequality~\ref{def-A3}, we have

\begin{align*}
A^3 &= \sum \limits_{j \not \in R^2_i(r_j)}  \left( \frac{19 w_j p_{ij}}{\epsilon} + w_j p_{i,\kappa_i(r_j)} + w_j \sum \limits_{h \in V_i(r_j^-): \delta_{ih} \geq \delta_{ij}} p_{ih} + w_j p_{ij} + p_{ij}\sum \limits_{h \in V_i(r_j^-): \delta_{ij} > \delta_{ih}} w_{ih} \right) \\
&\qquad - \sum \limits_{j \not \in R^2_i(r_j)} \left(p_{ij} \sum \limits_{h \in R^2_i(r_j)} w_h +  \epsilon^2 W_i(r_j) p_{ij} \right)\\ 
&\geq   \sum \limits_{j \not \in R^2_i(r_j)}  \left( \frac{19 w_j p_{ij}}{\epsilon} + w_j p_{i,\kappa_i(r_j)} + w_j \sum \limits_{h \in V_i(r_j^-): \delta_{ih} \geq \delta_{ij}} p_{ih} + w_j p_{ij} + p_{ij}\sum \limits_{h \in V_i(r_j^-): \delta_{ij} > \delta_{ih}} w_{ih} \right) \\
&\qquad - \sum \limits_{j \not \in R^2_i(r_j)} \left(p_{ij} (2\epsilon w_j + w_{\nu_i(r_j^-)}) +  \epsilon^2 W_i(r_j) p_{ij} \right)\\
&\text{~Follows from Property 3 in Lemma~\ref{wg-rule-prop}}\\
&\geq   \sum \limits_{j \not \in R^2_i(r_j)}  \left( \frac{17 w_j p_{ij}}{\epsilon} + w_j p_{i,\kappa_i(r_j)} + w_j \sum \limits_{h \in V_i(r_j^-): \delta_{ih} \geq \delta_{ij}} p_{ih} + w_j p_{ij} + p_{ij}\sum \limits_{h \in V_i(r_j^-): \delta_{ij} > \delta_{ih}} w_{ih} \right) \\
&\qquad - \sum \limits_{j \not \in R^2_i(r_j)} \left(p_{ij}  w_{\nu_i(r_j^-)}+  \epsilon^2 W_i(r_j) p_{ij} \right)\\ 
&\geq   \sum \limits_{j \not \in R^2_i(r_j)}  \left( \frac{17 w_j p_{ij}}{\epsilon} + w_j p_{i,\kappa_i(r_j)} + w_j \sum \limits_{h \in V_i(r_j^-): \delta_{ih} \geq \delta_{ij}} p_{ih} + w_j p_{ij} + p_{ij}\sum \limits_{h \in V_i(r_j^-): \delta_{ij} > \delta_{ih}} w_{ih} \right) \\
&\qquad - \sum \limits_{j \not \in R^2_i(r_j): w_j/\epsilon \leq w_{\nu_i(r_j^-)}} p_{ij}  w_{\nu_i(r_j^-)} - \sum \limits_{j \not \in R^2_i(r_j): w_j/\epsilon > w_{\nu_i(r_j^-)}} p_{ij}  w_{\nu_i(r_j^-)} - \sum \limits_{j \not \in R^2_{ij}}  \epsilon^2 W_i(r_j) p_{ij}\\ 
&\geq   \sum \limits_{j \not \in R^2_i(r_j)}  \left( \frac{17 w_j p_{ij}}{\epsilon} + w_j p_{i,\kappa_i(r_j)} + w_j \sum \limits_{h \in V_i(r_j^-): \delta_{ih} \geq \delta_{ij}} p_{ih} + w_j p_{ij} + p_{ij}\sum \limits_{h \in V_i(r_j^-): \delta_{ij} > \delta_{ih}} w_{ih} \right) \\
&\qquad - \sum \limits_{j \not \in R^2_i(r_j): w_j/\epsilon \leq w_{\nu_i(r_j)}} p_{ij}  w_{\nu_i(r_j)} - \sum \limits_{j \not \in R^2_{ij}: w_j/\epsilon > w_{\nu_i(r_j^-)}} p_{ij}  w_j/\epsilon - \sum \limits_{j \not \in R^2_{ij}}  \epsilon^2 W_i(r_j) p_{ij} \\ 
&\geq   \sum \limits_{j \not \in R^2_i(r_j)}  \left( \frac{16 w_j p_{ij}}{\epsilon} + w_j p_{i,\kappa_i(r_j)} + w_j \sum \limits_{h \in V_i(r_j^-): \delta_{ih} \geq \delta_{ij}} p_{ih} + w_j p_{ij} + p_{ij}\sum \limits_{h \in V_i(r_j^-) \setminus \{k_i(r_j^-)\}: \delta_{ij} > \delta_{ih}} w_{ih} \right) \\
&\qquad - \sum \limits_{j \not \in R^2_i(r_j): w_j/\epsilon \leq w_{\nu_i(r_j^-)}} p_{ij}  w_{\nu_i(r_j)}  - \sum \limits_{j \not \in R^2_{ij}} \epsilon^2 W_i(r_j) p_{ij} \\ 
&\geq   \sum \limits_{j \not \in R^2_i(r_j)}  \left( \frac{16 w_j p_{ij}}{\epsilon} + w_j p_{i,\kappa_i(r_j)} + w_j \sum \limits_{h \in V_i(r_j^-): \delta_{ih} \geq \delta_{ij}} p_{ih} + w_j p_{ij} + p_{ij}\sum \limits_{h \in V_i(r_j^-): \delta_{ij} > \delta_{ih}} w_{ih} \right) \\
&\qquad - \sum \limits_{j \not \in R^2_i(r_j): w_j/\epsilon \leq w_{\nu_i(r_j^-)}} w_{j}  p_{i,\nu_i(r_j^-)}  - \sum \limits_{j \not \in R^2_{ij}} \epsilon^2 W_i(r_j) p_{ij} 
		\qquad \text{since $\delta_{ij} \geq \delta_{i,\nu_i(r_j)}$} \\
&\geq   \sum \limits_{j \not \in R^2_i(r_j)}  \left( \frac{16 w_j p_{ij}}{\epsilon} + w_j p_{i,\kappa_i(r_j)} + w_j \sum \limits_{h \in V_i(r_j^-): \delta_{ih} \geq \delta_{ij}} p_{ih} + w_j p_{ij} + p_{ij}\sum \limits_{h \in V_i(r_j^-): \delta_{ij} > \delta_{ih}} w_{ih} \right) \\
&\qquad - \sum \limits_{j \not \in R^2_i(r_j): w_j/\epsilon \leq w_{\nu_i(r_j^-)}} \epsilon w_{\nu_i(r_j^-)}  p_{i,\nu_i(r_j^-)}  - \sum \limits_{j \not \in R^2_{ij}}  \epsilon^2 W_i(r_j) p_{ij} \\
&\geq   \sum \limits_{j \not \in R^2_i(r_j)}  \left( \frac{15 w_j p_{ij}}{\epsilon} + w_j p_{i,\kappa_i(r_j)} + w_j \sum \limits_{h \in V_i(r_j^-): \delta_{ih} \geq \delta_{ij}} p_{ih} + w_j p_{ij} + p_{ij}\sum \limits_{h \in V_i(r_j^-): \delta_{ij} > \delta_{ih}} w_{ih} \right) \\
&\qquad - \sum \limits_{j \not \in R^2_i(r_j)}  \epsilon^2 W_i(r_j) p_{ij} \\
&\geq   \sum \limits_{j \not \in R^2_i(r_j)}  \left( \frac{15 w_j p_{ij}}{\epsilon} + w_j p_{i,\kappa_i(r_j)} + w_j \sum \limits_{h \in V_i(r_j^-): \delta_{ih} \geq \delta_{ij}} p_{ih} + w_j p_{ij} + p_{ij}\sum \limits_{h \in V_i(r_j^-): \delta_{ij} > \delta_{ih}} w_{ih} \right) \\
&\qquad - \frac{5}{\epsilon} \left( \sum \limits_{j \in \mathcal{J} } w_j p_{ij} \right) \qquad \text{this follows from the Corollary~\ref{coro-bound-partof-objective}}\\
&\geq \sum \limits_{j \not \in R^2_i(r_j)}  \left(w_j (\widetilde{C}_j - r_j) + \frac{15 w_j p_{ij}}{\epsilon}\right) - \frac{5}{\epsilon} \left( \sum \limits_{j \in \mathcal{J} } w_j p_{ij} \right)
\end{align*}
\end{proof}

Combining Claims~\ref{lb-A1},~\ref{lb-A2} and~\ref{lb-A3}, we have that 
\begin{align*}
\epsilon/(1+\epsilon) \sum \limits_j \alpha_j = A^1 + A^2 + A^3 \geq \sum \limits_{j}  \left(w_j (\widetilde{C}_j - r_j) + \frac{10 w_j p_{ij}}{\epsilon} \right) 
\end{align*}

\end{proof}

\subsection{Proof of theorem~\ref{main-theorem}}
\begin{proof}
In the definition of dual variables, each job $j$ is accounted in $\beta_{it}$ variable until its \emph{definitive completion time}. Thus, $\sum \limits_{i,t} \beta_{it} \leq  \frac{\epsilon}{(1+\epsilon)(1+\epsilon^2)} \sum_{j \in \mathcal{J}}  w_j (\widetilde{C}_j - r_j)$.  Combining it with Lemma~\ref{bound-alpha}, we have that the dual objective is at least $\frac{\epsilon^3}{(1+\epsilon)(1+\epsilon^2)} \sum_{j \in \mathcal{J}}  w_j (\widetilde{C}_j - r_j)$. Further, the cost of the primal is at most $22 \sum_{j \in \mathcal{J}} w_j (\widetilde{C}_j - r_j)$. Hence the theorem follows. 
\end{proof}

\bibliography{references}

\end{document}